\documentclass{amsart}
\usepackage{amsfonts}
\allowdisplaybreaks
\usepackage{amsfonts}
\newtheorem{lem}{Lemma}

\newtheorem{thm}{Theorem}

 \usepackage{graphicx}
\usepackage{amsthm,amsmath,verbatim,amssymb}
\usepackage{color}
\usepackage{arydshln}
  \DeclareMathOperator{\Tr}{Tr}
 
\newcommand{\p}{\psi}
\newcommand{\la}{\lambda}
 \newcommand{\e}{\mathbb{E}}
   \newcommand{\R}{\mathbb{R}}
\newtheorem{rem}{Remark}

\begin{document}
\title[Spectrum]{Spectrum of SYK model II:\\ Central limit theorem }
\author[Feng, Tian and Wei]{Renjie Feng, Gang Tian, Dongyi Wei}
\address{Beijing International Center for Mathematical Research, Peking University, Beijing, China, 100871.}

\email{renjie@math.pku.edu.cn}
\email{gtian@math.pku.edu.cn}
\email{jnwdyi@pku.edu.cn}

\date{\today}
   \maketitle
    \begin{abstract}
 In our previous paper \cite{FTD1}, we derived the almost sure convergence of the global density of eigenvalues of random matrices of the SYK model. In this paper, we will prove the central limit theorem for the linear statistic of  eigenvalues and compute its variance. 
  \end{abstract}

\section{Introduction}
The SYK model is a random matrix model in the form of \cite{black, K, MS, SY}  \begin{equation}\label{sykm}H=i^{[q_n/2]}\frac{1}{\sqrt{{{n} \choose {q_n}}}}\sum_{1\leq i_1<i_2< \cdots < i_{q_n} \leq n} J_{i_1i_2\cdots i_{q_n}}{\p}_{i_1}{\p}_{i_2}\cdots {\p}_{ i_{q_n}},\end{equation}
where $J_{i_1i_2\cdots  i_{q_n}} $ are independent identically distributed
(i.i.d.) random variables with mean 0 and variance 1; we further assume that  the $k$-th moment  of $|J_{i_1i_2\cdots  i_{q_n}}| $ is uniformly bounded for any fixed $k$;
 $\p_j$ are Majorana fermions  satisfying the algebra
\begin{equation}\label{anti}\left\{\p_i,\p_j\right\}:=\p_i\p_j+\p_j\p_i=2\delta_{ij},\,\,\,1\leq i, j\leq n  .\end{equation}
Throughout the article,  $n$ is an even integer. As a remark,  physicists care especially when $q_n$ is an even integer, but  the model is still a good one if $q_n$ is odd from the mathematical point of view, and our main results apply to both cases. %And the SYK model is a very fast developing topic in high energy physics, we refer to for more physics background. 

  By the representation of the Clifford algebra, $\p_i$ can be represented by $L_n\times L_n$ Hermitian matrices with $L_n=2^{n/2}$ \cite{LM}, and thus $H$ is  also Hermitian. Let  $\la_{i}, 1\leq i \leq L_n$ be the eigenvalues of $H$ which are real numbers. Let's define the normalized empirical measure of eigenvalues of $H$  as 	\begin{equation}\label{emp}
  \rho_n(\lambda):=\frac 1{L_n}\sum_{j=1}^{L_n} \delta_{\la_{j}}(\la).
	\end{equation}

	In our first paper \cite{FTD1}, we proved that $\rho_n$
 converges to a  probability measure $\rho_\infty$ with probability 1 (or almost surely). Such result can be view as a type of  `law of large numbers' in probability theory.  Actually, let $q_n$ be even, the limiting density $\rho_\infty$ will depend on the limit of the quotient ${q_n^2}/n$ if $1\leq q_n\leq n/2$ or ${(n-q_n)^2}/n$ if $n/2 \leq q_n<n$. The results for odd $q_n$ are similar.  We refer to Theorems 1 and 2 in \cite{FTD1} for the precise statements. 
 
  In this paper, other than the  `law of large numbers', we will prove the central limit theorem (CLT) for the linear statistic of eigenvalues of the SYK model and compute its variance as $n\to\infty$. The CLT is one of the most important theorems in probability theory and random matrix theory. Our results indicate some useful information about the global 2-point correlation of eigenvalues, we also refer to  the recent papers \cite{GV3,GV4} for the numerical results on the local 
  behavior (or rescaling limit) of the 2-point correlation. 
  
  Given a test function $f(x)$,   the linear statistic of  eigenvalues is
  $$\mathcal L_n (f):=\langle f, \rho_n\rangle =\frac 1{L_n}\sum_{j=1}^{L_n} f(\lambda_j).$$
In random matrix theory,  the investigation of  the CLT for the linear statistic of
eigenvalues of random matrices dates back to Jonsson   on Gaussian
Wishart matrices \cite{J}. Similar work for the Wigner matrices  was derived by Sinai-Soshnikov \cite{Ss} and more general results by Johansson \cite{Jo}. There are many contexts on this, we refer to  \cite{LP, LSW}, Chapter 9 in \cite{BS}  and Chapter 3 in \cite{PS}  for  more details.

It's also worth mentioning the CLT for the linear statistic of many other random point processes, where basically  the variance of the linear statistic can be  expressed as
some energy functional of the test function.  Actually it's really hard to list all of these point processes, we   refer to the following results and the references therein: Sodin-Tsirelson's work on zeros of random polynomials and random analytic functions \cite{Sod}, Shiffman-Zelditch's work on zeros of random holomorphic sections over the complex manifolds \cite{SZ}, Berman's result regarding the Fekete points defined via the Bergman kernel on the complex manifolds \cite{Be} and Soshnikov's result on the determinantal point processes  \cite{SOA}. 
%One may also find the results  for other point processes  in \cite{HKMY,SOA} and the references therein.

In \S\ref{morevariance},  we will prove the following CLT for the general SYK model,
 \begin{thm}\label{main3} Let $J_{i_1i_2\cdots  i_{q_n}}$ be 
i.i.d. random variables with mean 0 and variance 1, and the $k$-th moment  of $|J_{i_1i_2\cdots  i_{q_n}}| $ is uniformly bounded. Let's denote $\gamma:=\e J_{i_1\cdots  i_{q_n}}^4$ as the 4-th moment of the random variable. Let $q_n$ be either even or odd integers.  Let $\rho_\infty$ be the
limiting density of eigenvalues of the SYK model as in Theorem 1 (if $q_n$ is even) or Theorem 2 (if $q_n$ is odd) in \cite{FTD1}, which also depends on the limit of the quotient ${q_n^2}/n$ if $1\leq q_n\leq n/2$ or ${(n-q_n)^2}/n$ if $n/2 \leq q_n< n$.  Let  $f(x)$ be a real polynomial. Then we  have the following convergence in distribution as $n\to\infty$,   $$  {n\choose q_n}^{\frac{1}{2}} (\mathcal L_n(f)-\e \mathcal L_n(f))\Rightarrow \langle xf'/2,  \rho_{\infty}\rangle J,$$ where $J$ is the Gaussian distribution with mean 0 and variance $ \gamma-1$. In particular,  the limit of its variance satisfies
 $$  \lim\limits_{n\to+\infty}{n\choose q_n} \text{var}[\mathcal L_n(f)]= \langle xf'/2,  \rho_{\infty} \rangle^2(\gamma-1).$$
\end{thm}

%As a direct consequence, we have 
\begin{rem}Theorem \ref{main3} implies the CLT for the trace $L_n^{-1}\Tr H^k$ for any fixed $k$, for example, let's take  $f(x)=x^2$, then we will have the CLT for $L_n^{-1}[\sum_{j=1}^{L_n}\la_j^2]$ which is $L_n^{-1} \Tr H^2$. In \cite{Ss}, Sinai-Soshnikov proved the CLT for  $\Tr W^{k_N}$ where $W$ are  general $N\times N$ Wigner matrices and $k_N$ are some slowly growing functions of $N$ (see  \cite{AGZ, BS, PS} also). As a consequence, they proved the CLT for analytic test functions on the disk and the almost sure convergence of the largest eigenvalue.  But for the SYK model, there are essential difficulties to prove such type of results. 
\end{rem}
 Actually, the proof of Theorem \ref{main3} is based on  careful estimates of the variance and covariance of the trace  $L_n^{-1}\Tr H^k$  (see \S\ref{morevariance}). The estimates will also imply that Theorem \ref{main3} holds for some class of analytic functions other than polynomials (see Remark \ref{redd}). 
  %The main problem is that we do not know if Theorem \ref{main3} holds for  more general class of functions, such as the smooth functions with compact support. But this is true for 
 In the special case of the Gaussian SYK model, we can 
improve  Theorem \ref{main3} to a larger class of functions.

 \begin{thm}\label{main3a}   For the Gaussian SYK model where $J_{i_1i_2\cdots  i_{q_n}}$ are
i.i.d. standard Gaussian random variables such that the fourth moment is $\gamma=3$, Theorem \ref{main3} holds for a class of functions $f(x)$,  where  $f(x)$ are   Lipschitz functions  and $f'(x)$ are bounded uniformly continuous.  %$$  {n\choose q_n}^{\frac{1}{2}} (\mathcal L_n(f)-\e \mathcal L_n(f))\Rightarrow \sqrt 2\langle xf'/2,  \rho_{\infty}\rangle \mathcal N,$$ where $\mathcal N$ is the standard Gaussian distribution and $\rho_\infty$ are the limiting densities   according to the limit $q^2_n/n$. %In particular,  the limit of the variance satisfies
% $$  \lim\limits_{n\to+\infty}{n\choose q_n} \text{var}[\mathcal L_n(f)]= \langle xf'/2,  \rho_{\infty} \rangle^2(\gamma-1).$$
\end{thm}

As a special case,  Theorem \ref{main3a} is true when test functions are smooth with compact support. But we  do not know what happen in general, especially when the test functions are not smooth enough or singular. Such cases  are intensively studied  in  random matrix theory, a good reference  is \cite{PS}. For example, there are two important types of test functions considered in random matrix theory:  when the test function is $\ln |x|$, then one may derive the CLT for the logarithmic determinant of the random matrices (see \cite{TV} for Tao-Vu's proof for general Wigner matrices); if we take the test function as the characteristic function  supported on some interval, then one may get the CLT for the number of eigenvalues falling in such interval (see Soshnikov's results for the determinantal point process which can be applied to the random matrices of GUE \cite{SOA}). For the  SYK model, even in the Gaussian case, we still do not know if the CLT holds for these two types of functions, we postpone these problems for further investigations.

Note that there is a symmetry between the systems with the interaction of $q_n$ fermions and  $n-q_n$ fermions (see \cite{FTD1}), therefore,  we  only prove the main theorems for  even $q_n$ with $2\leq q_n\leq n/2$,  the rest cases (even  $q_n$ with  $n/2\leq q_n<n$ or odd $q_n$) follow immediately without any essential difference and we omit the proof. % We   first prove Theorem \ref{main3} by studying the limit of the covariance of the trace  $L_n^{-1}\Tr H^k$. Then we prove Theorem \ref{main3a} by  approximations making use of the 
%F\'ejer kernel and Theorem \ref{main3}.    

%As a final remark, a related model is the quantum $q$-spin glass model considered in \cite{So, KLW, KLW2}, where the authors derived the limiting densities and proved that there is also a phase transition. Later on, the concentration of measure theorem of the $q$-spin glass model and the almost sure convergence of the normalized empirical measure to its limiting measure in $(M_1(\mathbb{R}), d_{BL}) $ are proved in \cite{BM}. \\

%\textbf{Acknowledgement:} The first named author would like to thank Gerard Ben Arous for many helpful discussions when he was visiting NYU Shanghai.

\section{Preliminary
}\label{notations}
\subsection{Notations and basic properties}\label{nota}
Let's first review some notations and basic properties in \cite{FTD1} that we will make use of in this paper. 

  For a set $A=\{i_1, i_2,\cdots, i_{m}\}$$  \subseteq\{1,2,\cdots,n\}$, $ 1\leq i_1<i_2< \cdots <i_m \leq n,$ we denote  $$\Psi_A:=\psi_{i_1}\cdots \psi_{i_{m}} \,\,\,\mbox{and}\,\,\,\Psi_A:=I\,\,\mbox{if}\,\,A=\emptyset.$$ We denote the set   $$I_n=\{(i_1, i_2,\cdots, i_{q_n}),\,\,1\leq i_1<i_2< \cdots <i_{q_n} \leq n \}.$$Thus the cardinality of $I_n$ is $$|I_n|={n\choose q_n}.$$ For any coordinate $R=(i_1, \cdots, i_{q_n})\in I_n$, we denote $$J_R:=J_{i_1\cdots i_{q_n}}\,\,\,\mbox{and}\,\,\,\,\Psi_R:=\psi_{i_1}\cdots \psi_{i_{q_n}}. $$ Sometimes we identify $R$ with the set $\{i_1, \cdots, i_{q_n}\}.$
  Thus we can simply rewrite the SYK model as 
\begin{equation}\label{sim}H=i^{[q_n/2]}\frac{1}{\sqrt{{{n} \choose {q_n}}}}\sum_{R\in I_n} J_{R}{\Psi}_{R}\end{equation}
Given any set $X$ and any integer $k\geq 1$, we define $P_2(X^k)$
 to be   the tuples $(x_1,\cdots, x_k) \in X^k$
for
which all entries $x_1,\cdots, x_k$ appear exactly twice. If $k$ is odd, then $P_2(X^k)$ is an empty set.

Throughout the article, we denote $c_k$
as some constant  depending only on $k$ and independent of $n$ and $q_n$, but its value may differ from line to line, the same  for $c_{2k}$, $c_k'$  and so forth.
 
We will also need  the following properties that one can find the proof in \cite{FTD1},
\begin{itemize}
\item  Given a set $A\subseteq\{1,2,\cdots,n\}$,  $$\mbox{$\Tr\Psi_A=0 $ and $\Psi_A\neq \pm I $ are always true for $A\neq \emptyset$}. $$ \item For $A,B\subseteq\{1,2,\cdots,n\}$, then  $$\mbox{$\Psi_A=\pm\Psi_B$ if and only if $A=B$.} $$
\item And $$\mbox{$\Psi_A\Psi_B=\pm\Psi_{A\triangle B}$  where $A\triangle B:=(A\setminus B)\cup (B\setminus A).$} $$ \item For $ {A_1},\cdots ,{A_k}\subseteq\{1,2,\cdots,n\}$, we have 
 \begin{equation}\label{tr}
|\frac 1 {L_n} \Tr \Psi_{A_1}\cdots \Psi_{A_k}|\leq 1.
 \end{equation} 
\end{itemize}
\subsection{Moments }\label{moments}
%Now we review some results about the moments of the probability measure $\rho_\infty$,  which is the limit of the empirical measure  of $\rho_n$ in \eqref{emp}.

Given any even integer $k$,  we define the set of $2$ to $1$ maps as \begin{equation}\label{sk}S_k=\left\{\pi: \,\,\{1,2,\cdots, k\}\to \{1,2,\cdots, \frac k2\}| |\pi^{-1}(j)|=2,  1\leq j\leq \frac k2  \right\}.\end{equation}
The crossing number $\kappa (\pi)$ for a pair-partition $\pi$ is defined to be the number of subsets
$\{r,s\}\subset  \{1,2,\cdots, \frac k2\}$ such that there exists $1\leq a<b<c<d\leq k, \pi(a)=\pi(c)=r,  \pi(b)=\pi(d)=s$. %Let $\{r_1,s_1\},\{r_2,s_2\},\cdots, \{r_{\kappa(\pi)},s_{\kappa(\pi)}\}$ be the crossings of $\pi$. 
 
Given $a>0$, throughout the article, we denote
\begin{equation}\label{mka} 
% \frac 1{L_n} \frac 1{{n\choose q_n}^{k/2}} \sum_{\pi\in S_k} \sum_{R_1,\cdots, R_{\frac k2}\in I_n, R_i \neq R_j \,\mbox{if}\,\,
%i\neq j}   \Tr \Psi_{R_{\pi(1)}}\cdots \Psi_{R_{\pi(k)}}\\
 m_k^a:=  \begin{cases} \frac{1}{(k/2)!}\sum_{\pi\in S_k} e^{-2a\kappa(\pi)}  &  \text{if $k$ is even,}\\
  0  &\text{if $k$ is odd.}  \end{cases} 
\end{equation}
It's further proved in \cite{ISV} that
\begin{equation}
% \frac 1{L_n} \frac 1{{n\choose q_n}^{k/2}} \sum_{\pi\in S_k} \sum_{R_1,\cdots, R_{\frac k2}\in I_n, R_i \neq R_j \,\mbox{if}\,\,
%i\neq j}   \Tr \Psi_{R_{\pi(1)}}\cdots \Psi_{R_{\pi(k)}}\\
 m_k^0:=\lim_{a\to 0}m_k^a=  \begin{cases}(k-1)!!  &  \text{if $k$ is even,}\\
  0  &\text{if $k$ is odd,}  \end{cases} 
\end{equation}
which is  the $k$-th moment of the standard Gaussian measure; and 
\begin{equation}
% \frac 1{L_n} \frac 1{{n\choose q_n}^{k/2}} \sum_{\pi\in S_k} \sum_{R_1,\cdots, R_{\frac k2}\in I_n, R_i \neq R_j \,\mbox{if}\,\,
%i\neq j}   \Tr \Psi_{R_{\pi(1)}}\cdots \Psi_{R_{\pi(k)}}\\
 m_k^\infty:=\lim_{a\to \infty }m_k^a=  \begin{cases}\frac {k!}{(k/2)!(k/2+1)!} &  \text{if $k$ is even,}\\
  0  &\text{if $k$ is odd,}  \end{cases} 
\end{equation}
which is the $k$-th moment of the semicircle law (or the Catalan numbers). 

Let $2\leq q_n\leq n/2$ be an even integer, 
then the main result proved  in \cite{ FTD1} is that, if $q_n^2/n\to a \in [0,+\infty]$, then the expectation of the $k$-th moment of the normalized empirical measure $\rho_n$ defined by \eqref{emp}  always satisfies
\begin{equation}m_{n, k}^a:=\e \langle x^k, \rho_n \rangle \to m_k^a,\,\,\, n\to\infty.\end{equation}
% The results for $q_n$  odd are similar.  %We will also need the above results in the next section. 

%Defining $m_k^a=0 $ for $k$ odd, then we have proved
%$$\lim_{n\to\infty} m_{n,k}=
   %m^a_k .$$

%By  Theorem \ref{main1}, we first have $m_k^a= \langle x^k,  \rho_{\infty} \rangle.$ 
%For convenience, we may define $m_k^a$ in \eqref{mka} for $a\in [0,\infty]$ where $m_k^0=(k-1)!!$ and $m_k^\infty=\frac {k!}{(k/2)!(k/2+1)!}$ if $k$ is even and both are $0$ if $k$ is odd.

  %\begin{equation}\label{mkd}
%\begin{aligned}m_{n,k}:&=\langle x^k, \mathbb E \rho_n \rangle =\frac 1{L_n}\e \sum_{i}(\la_j)^k=\frac 1{L_n}    \e [\Tr H^k] \\ &  \to m_k=  \begin{cases}
   %0  &\text{if $k$ is odd,} \\
    %( k   -1)!! &  \text{if $k$ is even.}
 %\end{cases} \end{aligned} 
 %\end{equation}
\section{Linear statistic and CLT}\label{morevariance}
In this section, we will prove Theorem \ref{main3} by analyzing the limit  of the covariance of $L_n^{-1}\Tr H^k$. %We  only prove the case when $q_n$ is even with  $2\leq q_n\leq n/2$, the results can be extended to even $q_n$ with $n/2\leq q_n< n$ and the odd case immediately.

\subsection{Limit of covariance}

%Now we compute the limit of the covariance of $L_n^{-1}\Tr H^k$.  

\begin{lem}\label{lemma11}Let $q_n$ be even and $k,k'\geq1,\ 2\leq q_n\leq n/2$. We assume $\ q_n^2/n\to a\in[0,+\infty]$ and denote the fourth moment $\gamma:=\e J_R^4$,  then we have $$  \lim\limits_{n\to+\infty}{n\choose q_n} \text{cov}(L_n^{-1} \Tr H^k,L_n^{-1} \Tr H^{k'})=(m_k^ak/2)(m_{k'}^ak'/2)(\gamma-1).$$
\end{lem}\begin{proof}

We first consider the case when $k+k'$ is even.  By \eqref{sim}, we have $$ \frac 1{L_n}    \Tr H^k= \frac 1{L_n} \frac {i^{q_nk/2}}{{n\choose q_n}^{k/2}} \sum_{{R_1},...,  {R_k}\in I_n} J_{R_1}\cdots J_{R_k}  \Tr \Psi_{R_1}\cdots \Psi_{R_k}, $$and thus \begin{align*}& \text{cov}(L_n^{-1} \Tr H^k,L_n^{-1} \Tr H^{k'})= \frac 1{L_n^2} \frac {i^{q_n(k+k')/2}}{{n\choose q_n}^{(k+k')/2}} \sum_{{R_1},...,  {R_{k+k'}}\in I_n}\\ &\text{cov}(J_{R_1}\cdots J_{R_k},J_{R_{k+1}}\cdots J_{R_{k+k'}})\cdot  \Tr \Psi_{R_1}\cdots \Psi_{R_k}\Tr \Psi_{R_{k+1}}\cdots \Psi_{R_{k+k'}}. \end{align*}For every ${R_1},...,  {R_{k+k'}}\in I_n $ and $A\in I_n$, let $ \#A=|\{{j}|1\leq j\leq k+k',R_j=A\}|.$

 If some $R_i$ appears only once in $(R_1,\cdots, R_{k+k'})$, then   $$\text{cov}(J_{R_1}\cdots J_{R_k},J_{R_{k+1}}\cdots J_{R_{k+k'}})=\e J_{R_1}\cdots J_{R_{k+k'}}-\e J_{R_1}\cdots J_{R_{k}}\e J_{R_{k+1}}\cdots J_{R_{k+k'}} =0 $$ by the independence of random variables.
Hence,  we can write\begin{align*}\text{cov}(L_n^{-1} \Tr H^k,L_n^{-1} \Tr H^{k'})=\text{cov}_1+\text{cov}_2\end{align*}where \begin{align*} \text{cov}_1&=  \frac 1{L_n^2} \frac {i^{q_n(k+k')/2}}{{n\choose q_n}^{(k+k')/2}} \sum_{({R_1},...,  {R_{k+k'}})\in P_2(I_n^{k+k'})} \text{cov}(J_{R_1}\cdots J_{R_k},J_{R_{k+1}}\cdots J_{R_{k+k'}})\\ &\cdot  \Tr \Psi_{R_1}\cdots \Psi_{R_k}\Tr \Psi_{R_{k+1}}\cdots \Psi_{R_{k+k'}},\end{align*} and \begin{align*}
	 \text{cov}_2&=  \frac 1{L_n^2} \frac {i^{q_n(k+k')/2}}{{n\choose q_n}^{(k+k')/2}} \sum_{({R_1},...,  {R_{k+k'}})\in I_n^{k+k'}\setminus P_2(I_n^{k+k'}), \# R_i\geq 2}\\ & \text{cov}(J_{R_1}\cdots J_{R_k},J_{R_{k+1}}\cdots J_{R_{k+k'}})\cdot  \Tr \Psi_{R_1}\cdots \Psi_{R_k}\Tr \Psi_{R_{k+1}}\cdots \Psi_{R_{k+k'}}. \end{align*} Let's first estimate $\text{cov}_1$.  For $({R_1},...,  {R_{k+k'}})\in P_2(I_n^{k+k'}) $,  we denote $A_1:=\{R_j|1\leq j\leq k\},\ A_2:=\{R_j|k+1\leq j\leq k+k'\}$ and $ A_0:=A_1\cap A_2$. Then we can decompose $$P_2(I_n^{k+k'})=\cup_{j=0}^2P_{2,j}^{(k,k')}(I_n^{k+k'}) $$ where\begin{align*} P_{2,0}^{(k,k')}(I_n^{k+k'})&=\{({R_1},...,  {R_{k+k'}})\in P_2(I_n^{k+k'})|A_0=\emptyset\},\\P_{2,1}^{(k,k')}(I_n^{k+k'})&=\{({R_1},...,  {R_{k+k'}})\in P_2(I_n^{k+k'})|A_0\neq\emptyset,\ \Psi_{R_1}\cdots \Psi_{R_k}=\pm I\},\\P_{2,2}^{(k,k')}(I_n^{k+k'})&=\{({R_1},...,  {R_{k+k'}})\in P_2(I_n^{k+k'})|\Psi_{R_1}\cdots \Psi_{R_k}\neq\pm I\}. \end{align*}
	If $({R_1},...,  {R_{k+k'}})\in P_{2,0}^{(k,k')}(I_n^{k+k'}) $, then $J_{R_1}\cdots J_{R_k}$ and $J_{R_{k+1}}\cdots J_{R_{k+k'}}$ are independent, hence $\text{cov}(J_{R_1}\cdots J_{R_k},J_{R_{k+1}}\cdots J_{R_{k+k'}})=0.$ If $({R_1},...,  {R_{k+k'}})\in P_{2,2}^{(k,k')}(I_n^{k+k'}) $, then $\Tr \Psi_{R_1}\cdots \Psi_{R_k}=0$ (see \S\ref{nota}). If $({R_1},...,  {R_{k+k'}})\in P_{2,1}^{(k,k')}(I_n^{k+k'}) $, then we easily have $\text{cov}(J_{R_1}\cdots J_{R_k},J_{R_{k+1}}\cdots J_{R_{k+k'}})=1$.  Thus, we have\begin{align*} |\text{cov}_1|&=\frac 1{L_n^2}  \frac {1}{{n\choose q_n}^{\frac{k+k'}{2}}} \left|\sum_{({R_1},...,  {R_{k+k'}})\in P_{2,1}^{(k,k')}(I_n^{k+k'})} \text{cov}(J_{R_1}\cdots J_{R_k},J_{R_{k+1}}\cdots J_{R_{k+k'}})\cdot\right.\\  &\left.\Tr \Psi_{R_1}\cdots \Psi_{R_k}\Tr \Psi_{R_{k+1}}\cdots \Psi_{R_{k+k'}}\right|\\ &\leq \frac 1{{n\choose q_n}^{\frac{k+k'}{2}}} \sum_{({R_1},...,  {R_{k+k'}})\in P_{2,1}^{(k,k')}(I_n^{k+k'})}1\\&= {n\choose q_n}^{-(k+k')/2}{|P_{2,1}^{(k,k')}(I_n^{k+k'})|},\end{align*}
	where we used inequality \eqref{tr}.
	
Now we estimate $ |P_{2,1}^{(k,k')}(I_n^{k+k'})|.$ Let $m=|A_0|>0,$ then there exists $1\leq i_1<\cdots<i_m\leq k$ and $k+1\leq i_1'<\cdots<i_m'\leq k+k'$ such that $A_0=\{R_{i_1},\cdots,R_{i_m}\}=\{R_{i_1'},\cdots,R_{i_m'}\}.$ Now we have $\Psi_{R_1}\cdots \Psi_{R_k}=\pm \Psi_{R_{i_1}}\cdots \Psi_{R_{i_m}} $ and \begin{align*} P_{2,1}^{(k,k')}(I_n^{k+k'})&=\{({R_1},...,  {R_{k+k'}})\in P_2(I_n^{k+k'})|m>0,\ \Psi_{R_{i_1}}\cdots \Psi_{R_{i_m}}=\pm I\}. \end{align*}Since $|A_1\cup A_2|=(k+k')/2,$ for every fixed $A_0$ there are $ {|I_n|-m\choose (k+k')/2-m}$ choices of $ A_1\cup A_2;$ for every fixed $A_1\cup A_2,$ there are at most $ ((k+k')/2)^{k+k'}$ choices of $({R_1},...,  {R_{(k+k')/2}}). $ Let's denote \begin{equation}\label{bm}B_m=\{({R_1},...,  {R_{m}})\in I_n^{m}|\Psi_{R_1}\cdots \Psi_{R_m}=\pm I,\ R_i\neq R_j,\ \forall\ 1\leq i<j\leq m\}.\end{equation}   Then for fixed $m$, every $A_0$ corresponds to exactly $m!$ elements in $B_m,$ thus the number of elements in $P_{2,1}^{(k,k')}(I_n^{k+k'}) $ satisfying $|A_0|=m$ is at most ${|I_n|-m\choose (k+k')/2-m}(\frac{k+k'}{2})^{k+k'}/(m!)\cdot|B_m| $.
%there exists $1\leq i_1<\cdots<i_m\leq k$ and $k+1\leq i_1'<\cdots<i_m'\leq k+k'$ such that $A_0=\{R_{i_1},\cdots,R_{i_m}\}=\{R_{i_1'},\cdots,R_{i_m'}\}.$ Now we have $\Psi_{R_1}\cdots \Psi_{R_k}=\pm \Psi_{R_{i_1}}\cdots \Psi_{R_{i_m}} $ and \begin{align*} P_{2,1}^{(k,k')}(I_n^{k+k'})&=\{({R_1},...,  {R_{k+k'}})\in P_2(I_n^{k+k'})|m>0,\ \Psi_{R_{i_1}}\cdots \Psi_{R_{i_m}}=\pm I\}. \end{align*}Since $|A_1\cup A_2|=(k+k')/2,$ for every fixed $A_0$ there are $ {|I_n|-m\choose (k+k')/2-m}$ choices of $ A_1\cup A_2;$ for every fixed $A_1\cup A_2,$ there are at most $ ((k+k')/2)^{k+k'}$ choices of $({R_1},...,  {R_{(k+k')/2}}). $   then for fixed $m$, every $A_0$ corresponds to exactly $m!$ elements in $B_m,$
By the estimate of $|B_m|$ in Lemma \ref{lemma10} below, we  will have
\begin{align*} |P_{2,1}^{(k,k')}(I_n^{k+k'})|&\leq \sum_{0<m\leq (k+k')/2}{|I_n|-m\choose (k+k')/2-m}\left(\frac{k+k'}{2}\right)^{k+k'}/(m!)\cdot|B_m|\\&\leq\sum_{0<m\leq (k+k')/2}c_{k+k',m}|I_n|^{(k+k')/2-m}\cdot|B_m|\\&\leq\sum_{0<m\leq (k+k')/2}c_{k+k',m}|I_n|^{(k+k')/2-m}c|I_n|^{m-1}n^{-\frac{1}{2}}\\&=c_{k+k'}|I_n|^{(k+k')/2-1}n^{-\frac{1}{2}}. \end{align*} Thus we have \begin{align*} |\text{cov}_1|&\leq {n\choose q_n}^{-(k+k')/2}{|P_{2,1}^{(k,k')}(I_n^{k+k'})|}\\ &\leq c_{k+k'}{n\choose q_n}^{-(k+k')/2}|I_n|^{(k+k')/2-1}n^{-\frac{1}{2}}\\&\leq c_{k+k'}{n\choose q_n}^{-1}n^{-\frac{1}{2}}. \end{align*}Hence, we have $$\lim\limits_{n\to \infty}{n\choose q_n}\text{cov}_{1}=0.$$ Therefore,  in order to prove Lemma \ref{lemma11},  we only need to prove  \begin{equation}\label{time}\lim\limits_{n\to \infty}{n\choose q_n}\text{cov}_{2}=(m_k^ak/2)(m_{k'}^ak'/2)(\gamma-1).\end{equation} If $({R_1},...,  {R_{k+k'}})\in I^{k+k'}_n\setminus P_2(I_n^{k+k'})$ and $\# R_i\geq 2 $, then $|\{R_j|1\leq j\leq k+k'\}|\leq (k+k')/2-1$. If the equality holds, then there are two possibilities.
 Type 1: some $R_i$ appears 4 times and all the rest appear exactly twice.
 Type 2: two distinct $R_i$ appear 3 times and all the rest appear twice. We denote $Q_{j}(I_n^{k+k'}) $ as the set of $({R_1},...,  {R_{k+k'}}) $ with Type $j$ for $j=1,2$ and $$Q_{3}(I_n^{k+k'})=\{({R_1},...,  {R_{k+k'}})\in I^{k+k'}_n:|\{R_j|1\leq j\leq k+k'\}|\leq (k+k')/2-2\}.$$
 Then we have $$\{({R_1},...,  {R_{k+k'}})\in I^{k+k'}_n\setminus P_2(I_n^{k+k'}): \# R_i\geq 2,\ \forall\ i\}=\cup_{j=1}^3Q_{j}(I_n^{k+k'}) $$ and we can further decompose   $$\text{cov}_2=\text{cov}_{2,1}+\text{cov}_{2,2}+\text{cov}_{2,3}$$ where \begin{align*} \text{cov}_{2,j}= \frac 1{L_n^2} \frac {i^{q_n(k+k')/2}}{{n\choose q_n}^{(k+k')/2}} \sum_{({R_1},...,  {R_{k+k'}})\in Q_j(I_n^{k+k'})} \text{cov}(J_{R_1}\cdots J_{R_k},J_{R_{k+1}}\cdots J_{R_{k+k'}})\\ \cdot  \Tr \Psi_{R_1}\cdots \Psi_{R_k}\Tr \Psi_{R_{k+1}}\cdots \Psi_{R_{k+k'}},\ \ \ j=1,2,3. \end{align*}
If $({R_1},...,  {R_{k+k'}})\in Q_2(I_n^{k+k'}), $ we assume $A$ and $B$ appear 3 times, $A\neq B$ and all the rest appear twice. Then by properties in \S \ref{nota} again, we have $\Psi_{R_1}\cdots \Psi_{R_{k+k'}}=\pm\Psi_{A}\Psi_{B}=\pm\Psi_{A\triangle B}\neq \pm I $, thus $\Psi_{R_1}\cdots \Psi_{R_k}\neq \pm I $ or $\Psi_{R_{k+1}}\cdots \Psi_{R_{k+k'}}\neq \pm I, $ this implies $\Tr\Psi_{R_1}\cdots \Psi_{R_k}=0 $ or $ \Tr\Psi_{R_{k+1}}\cdots \Psi_{R_{k+k'}}=0. $ Hence, we have identity   $${n\choose q_n}\text{cov}_{2,2}=0.$$ 

Let's denote $k_1:=(k+k')/2,$ then for $n$ large enough, we have \begin{align*} |Q_{3}(I_n^{k+k'})|&\leq\sum_{B\subseteq I_n,|B|=k_1-2}|\{({R_1},...,  {R_{k+k'}})\in I^{k+k'}_n|  R_i\in B,\ \forall\ 1\leq i\leq k+k' \}|\\ &=\sum_{B\subseteq I_n,|B|=k_1-2}({k_1-2})^{k+k'}\\&=({k_1-2})^{k+k'}{|I_n|\choose { {\frac{k+k'}{2}-2}}}\\&\leq c_{k+k'} |I_n|^{ {(k+k')/2-2}}. \end{align*} Thus we have\begin{align*}|\text{cov}_{2,3}|&\leq \frac 1{{n\choose q_n}^{\frac{k+k'}{2}}} \sum_{({R_1},...,  {R_{k+k'}})\in Q_{3}(I_n^{k+k'})} c_{k+k'}\\&\leq  \frac {c_{k+k'}|I_n|^{\frac{k+k'}{2}-2}}{{n\choose q_n}^{\frac{k+k'}{2}}} =c_{k+k'}{n\choose q_n}^{-2}.\end{align*}
Thus we have $${n\choose q_n}\text{cov}_{2,3}\to 0$$ which has no contribution to the left hand side of \eqref{time}. %And thus
%\begin{equation}\label{ddss}\lim_{n\to\infty}{n\choose q_n}\text{cov}_{2} =\lim_{n\to\infty}{n\choose q_n}\text{cov}_{2,1} .\end{equation}

Al of the rest effort is to estimate the last term $\text{cov}_{2,1}$.  For $({R_1},...,  {R_{k+k'}})\in Q_1(I_n^{k+k'}) $, we assume $A$ appears 4 times and all the rest appear exactly twice. Let's denote $A_1:=\{R_j|1\leq j\leq k\},\ A_2:=\{R_j|k+1\leq j\leq k+k'\},\ A_0:=A_1\cap A_2,\ k_0:=|\{j|1\leq j\leq k,R_j=A\}|,\ A_0^*=A_0\setminus\{A\}$ for $k_0$ even and $A_0^*=A_0$ for $k_0$ odd, then we can further decompose $$Q_1(I_n^{k+k'})=\cup_{j=0}^3Q_{1,j}^{(k,k')}(I_n^{k+k'}) $$ where\begin{align*} Q_{1,0}^{(k,k')}(I_n^{k+k'})&=\{({R_1},...,  {R_{k+k'}})\in Q_1(I_n^{k+k'})|A_0=\emptyset\},\\Q_{1,1}^{(k,k')}(I_n^{k+k'})&=\{({R_1},...,  {R_{k+k'}})\in Q_1(I_n^{k+k'})|A_0=\{A\},\ k_0=2\},\\Q_{1,2}^{(k,k')}(I_n^{k+k'})&=\{({R_1},...,  {R_{k+k'}})\in Q_1(I_n^{k+k'})|A_0^*\neq\emptyset,\ \Psi_{R_1}\cdots \Psi_{R_k}=\pm I\},\\Q_{1,3}^{(k,k')}(I_n^{k+k'})&=\{({R_1},...,  {R_{k+k'}})\in Q_1(I_n^{k+k'})|\Psi_{R_1}\cdots \Psi_{R_k}\neq\pm I\}. \end{align*}
As before, if $({R_1},...,  {R_{k+k'}})\in Q_{1,0}^{(k,k')}(I_n^{k+k'}) $, then $J_{R_1}\cdots J_{R_k}$ and $J_{R_{k+1}}\cdots J_{R_{k+k'}}$ are independent, hence $\text{cov}(J_{R_1}\cdots J_{R_k},J_{R_{k+1}}\cdots J_{R_{k+k'}})=0.$ If $({R_1},...,  {R_{k+k'}})\in Q_{1,3}^{(k,k')}(I_n^{k+k'}) $ then $\Tr \Psi_{R_1}\cdots \Psi_{R_k}=0. $ Therefore, we have $$\text{cov}_{2,1}=\text{cov}_{2,1,1}+\text{cov}_{2,1,2}$$ where \begin{align*} \text{cov}_{2,1,j}= \frac 1{L_n^2} \frac {i^{q_n(k+k')/2}}{{n\choose q_n}^{(k+k')/2}} \sum_{({R_1},...,  {R_{k+k'}})\in Q_{1,j}^{(k,k')}(I_n^{k+k'})} \text{cov}(J_{R_1}\cdots J_{R_k},J_{R_{k+1}}\cdots J_{R_{k+k'}})\\ \cdot  \Tr \Psi_{R_1}\cdots \Psi_{R_k}\Tr \Psi_{R_{k+1}}\cdots \Psi_{R_{k+k'}},\ \ \ j=1,2. \end{align*}
Let's  recall the following estimate proved in \cite{FTD1},
\begin{lem}\label{lemmma}Let $q_n$ be even,  for any $ k\geq1$, we have 
\begin{equation}\label{upp} \text{var}[L_n^{-1} \Tr H^k]\leq c_k{n\choose q_n}^{-1}\end{equation}where $c_k$ is some constant. \end{lem} By Lemma \ref{lemmma}, we easily have the upper bound
$$|\text{cov}(J_{R_1}\cdots J_{R_k},J_{R_{k+1}}\cdots J_{R_{k+k'}})|\leq c_{k+k'}.$$ Thus, we have\begin{align*} |\text{cov}_{2,1,2}|& \leq \frac 1{{n\choose q_n}^{(k+k')/2}} \sum_{({R_1},...,  {R_{k+k'}})\in Q_{1,2}^{(k,k')}(I_n^{k+k'})}c_{k+k'}\\&= {n\choose q_n}^{-k_1}c_{k+k'}{|Q_{1,2}^{(k,k')}(I_n^{k+k'})|}. \end{align*}Let $m=|A_0^*|>0,$ then there exists $1\leq i_1<\cdots<i_m\leq k$ and $k+1\leq i_1'<\cdots<i_m'\leq k+k'$ such that $A_0^*=\{R_{i_1},\cdots,R_{i_m}\}=\{R_{i_1'},\cdots,R_{i_m'}\}.$ Now we have $$\Psi_{R_1}\cdots \Psi_{R_k}=\pm \Psi_{R_{i_1}}\cdots \Psi_{R_{i_m}} $$ and \begin{align*} Q_{1,2}^{(k,k')}(I_n^{k+k'})&=\{({R_1},...,  {R_{k+k'}})\in Q_1(I_n^{k+k'})|m>0,\ \Psi_{R_{i_1}}\cdots \Psi_{R_{i_m}}=\pm I\}. \end{align*}
Following the same argument as $P_{2,1}(k,k')(I_n^{k+k'})$ above, we will have \begin{align*} |Q_{1,2}^{(k,k')}(I_n^{k+k'})|&\leq\sum_{0<m\leq k_1}(k_1-1)^{k+k'}{|I_n|-m\choose k_1-m-1}/(m!)\cdot|B_m|\\&\leq\sum_{0<m\leq k_1}(k_1-1)^{k+k'}{|I_n|-m\choose k_1-m-1}/(m!)\cdot|I_n|^{m-1}\\&\leq\sum_{0<m\leq k_1}c_{k+k',m}|I_n|^{k_1-m-1}|I_n|^{m-1}=c_{k+k'}|I_n|^{k_1-2}. \end{align*}
Hence, we have \begin{align*} |\text{cov}_{2,1,2}| &\leq {n\choose q_n}^{-k_1}c_{k+k'}{|Q_{1,2}^{(k,k')}(I_n^{k+k'})|}\\&\leq{n\choose q_n}^{-k_1}c_{k+k'}|I_n|^{k_1-2}\leq c_{k+k'}{n\choose q_n}^{-2} ,\end{align*}
i.e., $$ |{n\choose q_n}\text{cov}_{2,1,2}| \to0 ,$$
which implies that 
\begin{equation}\label{ddss}\lim_{n\to\infty}{n\choose q_n}\text{cov}_{2} =\lim_{n\to\infty}{n\choose q_n}\text{cov}_{2,1,1} .\end{equation}
 Now we estimate $\text{cov}_{2,1,1}$  which turns out to be  an interesting term.  

Recall the assumption that $k+k'$ is even in the beginning of the proof, for the case when $k$ and $k'$ are both odd,  by definition, $Q_{1,1}^{(k,k')}(I_n^{k+k'})$  must be empty, and thus $\text{cov}_{2,1,1}=0$. Now we discuss the case when  $k$ and $k'$ are both even.

By definition, given $({R_1},...,  {R_{k+k'}})\in Q_{1,1}^{(k,k')}(I_n^{k+k'}) $ with $A_0=A_1\cap A_2=\{A\}$ where $ A$ appears twice in both $({R_1},...,  {R_{k}})$ and $({R_{k+1}},...,  {R_{k+k'}})$,   we must  have $({R_1},...,  {R_{k}})\in P_{2}(I_n^{k})$ and $({R_{k+1}},...,  {R_{k+k'}}) \in P_{2}(I_n^{k'}) $. Furthermore, we have
$$\e[J_{R_1}\cdots J_{R_{k+k'}}]=\e[J_A^4]=\gamma$$ and \begin{align*}& \text{cov}(J_{R_1}\cdots J_{R_k},J_{R_{k+1}}\cdots J_{R_{k+k'}})\\&=\e[J_{R_1}\cdots J_{R_{k+k'}}]-\e[J_{R_1}\cdots J_{R_k}]\cdot\e[J_{R_{k+1}}\cdots J_{R_{k+k'}}]\\&=\gamma-1.\end{align*} Hence, we have\begin{align*} \text{cov}_{2,1,1}= \frac {\gamma-1}{L_n^2} \frac {i^{q_n(k+k')/2}}{{n\choose q_n}^{(k+k')/2}} \sum_{({R_1},...,  {R_{k+k'}})\in Q_{1,1}^{(k,k')}(I_n^{k+k'})}   \Tr \Psi_{R_1}\cdots \Psi_{R_k}\Tr \Psi_{R_{k+1}}\cdots \Psi_{R_{k+k'}}. \end{align*}
We first consider the case $a\in(0,+\infty).$ By definition \eqref{sk}, we can rewrite \begin{align*} \text{cov}_{2,1,1}= \frac {\gamma-1}{L_n^2} \frac {i^{q_n(k+k')/2}}{{n\choose q_n}^{(k+k')/2}} \sum_{\pi\in S_k}\sum_{\pi'\in S_{k'}}\sum_{R_1,\cdots, R_{(k+k')/2-1}\in I_n, R_i \neq R_j \,\mbox{if}\,\, i\neq j} \\  \frac{\Tr \Psi_{R_{\pi(1)}}\cdots \Psi_{R_{\pi(k)}}\Tr \Psi_{R_{\pi'(1)+k/2-1}}\cdots \Psi_{R_{\pi'(k)+k/2-1}}}{(k/2-1)!(k'/2-1)!}. \end{align*}
By the anticommutative relation \eqref{anti}, for any fixed $\pi$, we easily have (see \cite{FTD1}) \begin{equation} \label{lemma3}\frac {i^{q_nk/2}} {L_n} \Tr  \Psi_{R_{\pi(1)}}\cdots \Psi_{R_{\pi(k)}}=(-1)^{\sum_{k=1}^{\kappa(\pi)}|R_{r_k}\cap R_{s_k}|}.\end{equation} 
 
We also need the following lemma dealing with the cardinality of the intersection of the coordinates $|R_{r_k}\cap R_{s_k}|$ \cite{So}.
 \begin{lem}\label{lemma4}When $ q_n^2/n\to a\in (0,\infty)$, if we choose $\{R_1,\cdots, R_{\frac k2}\}$ uniformly from $I_n^{\frac k 2}$ with $R_i\neq R_j$ if $i\neq j$, then the intersection numbers $|R_{r_k}\cap R_{s_k}|,  k=1,\cdots , \kappa(\pi)$ are  approximately independently Poisson(a) distributed.  Here, $\{ r_k, s_k\}_{k=1}^{\kappa(\pi)}$ are crossings of $\pi$. \end{lem}
With indentity \eqref{lemma3} and Lemma \ref{lemma4}, for any fixed  map $\pi\in S_k,\pi'\in S_{k'}$, let $\{r_1,s_1\}, \{r_2,s_2\}, \cdots,  \{r_{\kappa(\pi)},s_{\kappa(\pi)}\}$ be the crossings of $\pi$ and $\{r_1',s_1'\},\{r_2',s_2'\},\cdots,\\ \{r_{\kappa(\pi')}',s_{\kappa(\pi')}'\}$ be the crossings of $\pi'$, then we have  \begin{align*}&\lim_{n\to \infty}  \frac 1{L_n^2} \frac{i^{q_n(k+k')/2}}{{n\choose q_n}^{(k+k')/2-1}} \sum_{R_1,\cdots, R_{(k+k')/2-1}\in I_n, R_i \neq R_j \,\mbox{if}\,\, i\neq j}\\&   \Tr \Psi_{R_{\pi(1)}}\cdots \Psi_{R_{\pi(k)}}\Tr \Psi_{R_{\pi'(1)+k/2-1}}\cdots \Psi_{R_{\pi'(k)+k/2-1}}\\
&=\lim_{n\to \infty}   \frac 1{{n\choose q_n}^{(k+k')/2-1}} \sum_{R_1,\cdots, R_{(k+k')/2-1}\in I_n, R_i \neq R_j \,\mbox{if}\,\, i\neq j}\\&  (-1)^{\sum_{k=1}^{\kappa(\pi)}|R_{r_k}\cap R_{s_k}|}(-1)^{\sum_{k=1}^{\kappa(\pi')}|R_{r_k'+k/2-1}\cap R_{s_k'+k/2-1}|} \\
&=\sum_{m_i\geq 0, 1\leq i\leq \kappa(\pi)}(-1)^{m_1+\cdots+m_{\kappa(\pi)}}\frac{a^{m_1+\cdots+m_{\kappa(\pi)}}}{m_1!\cdots m_{\kappa(\pi)}!}e^{-a\kappa(\pi)}\cdot\\ &\sum_{m_i\geq 0, 1\leq i\leq \kappa(\pi')}(-1)^{m_1+\cdots+m_{\kappa(\pi')}}\frac{a^{m_1+\cdots+m_{\kappa(\pi')}}}{m_1!\cdots m_{\kappa(\pi')}!}e^{-a\kappa(\pi')}\\&=e^{-2a\kappa(\pi)}e^{-2a\kappa(\pi')}.
\end{align*}Therefore, by definition \eqref{mka}, we will get
\begin{align}\label{mak2}\lim_{n\to \infty}{n\choose q_n}\text{cov}_{2,1,1}
% \frac 1{L_n} \frac 1{{n\choose q_n}^{k/2}} \sum_{\pi\in S_k} \sum_{R_1,\cdots, R_{\frac k2}\in I_n, R_i \neq R_j \,\mbox{if}\,\,
%i\neq j}   \Tr \Psi_{R_{\pi(1)}}\cdots \Psi_{R_{\pi(k)}}\\
 &= \frac{\gamma-1}{(k/2-1)!(k'/2-1)!}\sum_{\pi\in S_k} e^{-2a\kappa(\pi)}\sum_{\pi'\in S_{k'}} e^{-2a\kappa(\pi')}\\ \nonumber&=(m_k^ak/2)(m_{k'}^ak'/2)(\gamma-1).
\end{align} Actually,  \eqref{mak2} is also true if $k$ and $k'$ are both odd, since $\text{cov}_{2,1,1}$ and $m_k^a$ are both 0 for such case.  The above arguments making use of the crossing numbers still work for the case $a=0$. Therefore,   we  prove \eqref{time} when $a\in[0,\infty)$ and $k+k'$ is even.
% , since $\text{cov}_2=\text{cov}_{2,1}+\text{cov}_{2,3}=\text{cov}_{2,1,1}+\text{cov}_{2,1,2}+\text{cov}_{2,3}$, by all the estimates we derived above,  we  prove \eqref{time}.
% \begin{align*}\lim_{n\to \infty}{n\choose q_n}\text{cov}_{2,1,2}=\lim_{n\to \infty}{n\choose q_n}\text{cov}_{2,3}=0,\\ \lim_{n\to \infty}{n\choose q_n}\text{cov}_{2,1,1}=(m_k^ak/2)(m_{k'}^ak'/2)(\gamma-1)
%\end{align*}
%To summarize, we have proved
%$$\lim_{n\to \infty}{n\choose q_n}\text{cov}_{2}=(m_k^ak/2)(m_{k'}^ak'/2)(\gamma-1).$$

 The above arguments do not work for the case for $a=+\infty$, but we can use Lemma 5 in \cite{FTD1} to conclude that if $k+k'$ is  even, we still have \begin{align*}\lim_{n\to \infty}{n\choose q_n}\text{cov}_{2}=\lim_{n\to \infty}{n\choose q_n}\text{cov}_{2,1,1}=(m_k^\infty k/2)(m_{k'}^\infty k'/2)(\gamma-1).
\end{align*}% Actually, the above limit is also true for the case when $k$ and $k'$ are both odd. %then $m_k^a=0, $ and $$  \lim\limits_{n\to+\infty}{n\choose q_n} \text{var}[L_n^{-1} \Tr H^k]=(m_k^ak/2)^2(\gamma-1)=0.$$

 To summarize, combining the estimate of $\text{cov}_1$ and $\text{cov}_2$, for $k+k'$ even and  $a\in [0,\infty]$, we finally prove $$  \lim\limits_{n\to+\infty}{n\choose q_n} \text{cov}(L_n^{-1} \Tr H^k,L_n^{-1} \Tr H^{k'})=(m_k^ak/2)(m_{k'}^ak'/2)(\gamma-1).$$
 In particular, for any $k$, we have   $$  \lim\limits_{n\to+\infty}{n\choose q_n} \text{var}[L_n^{-1} \Tr H^k]=(m_k^ak/2)^2(\gamma-1).$$
In the end, if $k$ is odd and $k'$ is even, we  have\begin{align*} \left|{n\choose q_n} \text{cov}(L_n^{-1} \Tr H^k,L_n^{-1} \Tr H^{k'})\right|^2 &\leq{n\choose q_n} \text{var}[L_n^{-1} \Tr H^k]{n\choose q_n} \text{var}[L_n^{-1} \Tr H^{k'}]\\ & \to  (m_k^ak/2)^2(m_{k'}^ak'/2)^2(\gamma-1)^2=0,\end{align*}
since $m_k^a=0$ when $k$ is odd by definition. 

 Therefore, for $k+k'$ odd, we have $$  \lim\limits_{n\to+\infty}{n\choose q_n} \text{cov}(L_n^{-1} \Tr H^k,L_n^{-1} \Tr H^{k'})=0=(m_k^ak/2)(m_{k'}^ak'/2)(\gamma-1).$$  This completes the proof except the estimate of $|B_m|$.\end{proof}

 Now we prove the following technical lemma on the estimate of $|B_m|$ to  finish the proof of Lemma \ref{lemma11}.
\begin{lem}\label{lemma10} Let  $$B_m=\{({R_1},...,  {R_{m}})\in I_n^{m}|\Psi_{R_1}\cdots \Psi_{R_m}=\pm I,\ R_i\neq R_j,\ \forall\ 1\leq i<j\leq m\}.$$ Then  we have the estimate $$|B_m|\leq c|I_n|^{m-1}n^{-\frac{1}{2}},$$  where $c$ is an absolute constant independent of $m,n,q_n.$
\end{lem}
\begin{proof}
By definition we have $B_1=B_2=\emptyset,$ and we only need to consider the case $m\geq 3$. Let $ B_m^*=\{({R_1},...,  {R_{m}})\in I_n^{m}|\Psi_{R_1}\cdots \Psi_{R_m}=\pm I\}$, then we have $ B_m\subseteq B_m^*$ and  $|B_m|\leq |B_m^*|$. And we need to estimate $B(m,n,q_n)=|B_m^*|/|I_n|^{m-1}.$

{\bf Case 1:} $m=3.$ We have $ B_m^*=\{({R_1},R_2,  {R_{3}})\in I_n^{m}|\Psi_{R_1}\Psi_{R_2}=\pm \Psi_{R_3}\}=\{({R_1},R_2,  {R_{3}})\in I_n^{m}|\Psi_{R_1\triangle R_2}=\pm \Psi_{R_3}\}=\{({R_1},R_2,  {R_{3}})\in I_n^{m}|R_1\triangle R_2={R_3}\}.$ If $({R_1},R_2,  {R_{3}})\in B_m^*$, then $|R_3|=|R_1\triangle R_2|=|R_1|+| R_2|-2|R_1\cap R_2|$ and $|R_1|=|R_2|=|R_3|=q_n,$ thus $|R_1\cap R_2|=q_n/2.$ There are ${n\choose q_n}$ choices of $R_1$. For every fixed $R_1$, there are ${q_n\choose q_n/2}{n-q_n\choose q_n/2}$ choices of $R_2$ satisfying $|R_1\cap R_2|=q_n/2$. $R_3$ is uniquely determined by $R_1$ and $R_2$. Therefore, $|B_m^*|={n\choose q_n}{q_n\choose q_n/2}{n-q_n\choose q_n/2}$ and $B(3,n,q_n)=|B_m^*|/|I_n|^{2}={n\choose q_n}{q_n\choose q_n/2}{n-q_n\choose q_n/2}/{n\choose q_n}^2={q_n\choose q_n/2}{n-q_n\choose q_n/2}/{n\choose q_n}$.

 If $q_n$ is odd or $n<3q_n/2$,  then $B(3,n,q_n)=|B_m^*|=0$. If $q_n$ is even and $n>3q_n/2$, then $B(3,n,q_n)/B(3,n-1,q_n)={n-q_n\choose q_n/2}/{n-1-q_n\choose q_n/2}\cdot{n-1\choose q_n}/{n\choose q_n}=\frac{n-q_n}{n-3q_n/2}\frac{n-q_n}{n}=1-\frac{q_n(n-2q_n)}{(2n-3q_n)n}$ (Notice that the expression of $B(3,n,q_n)$ is well defined for every positive integer $n$). Thus for fixed $q_n,$ $B(3,n,q_n)$ is increasing for $3q_n/2\leq n\leq 2q_n$ and decreasing for $ n\geq 2q_n,$ which implies $B(3,n,q_n)\leq B(3,2q_n,q_n).$ We notice that\begin{align*} B(3,2q_n,q_n)&={q_n\choose q_n/2}^2/{2q_n\choose q_n}=\frac{(q_n!)^4}{((q_n/2)!)^4(2q_n)!}=\prod_{j=1}^{q_n/2}\frac{(2j)^4(2j-1)^4}{j^4(4j)(4j-1)(4j-2)(4j-3)}\\
&=\prod_{j=1}^{q_n/2}\frac{2(2j-1)^3}{j(4j-1)(4j-3)}=\prod_{j=1}^{q_n/2}\frac{2j-1}{2j}\left(1-\frac{1}{4(2j-1)^2}\right)^{-1}\\ &\leq
\prod_{j=1}^{q_n/2}\frac{2j-1}{2j}\left(1-\frac{1}{4j^2}\right)^{-1}=\prod_{j=1}^{q_n/2}\frac{2j}{2j+1} \leq
\prod_{j=1}^{q_n/2}\left(\frac{j}{j+1}\right)^{\frac{1}{2}}\\&=(q_n/2+1)^{-\frac{1}{2}}. \end{align*}Therefore, if $3q_n/2\leq n\leq 3q_n$, then \begin{align*} B(3,n,q_n)\leq B(3,2q_n,q_n)\leq(q_n/2+1)^{-\frac{1}{2}}\leq cn^{-\frac{1}{2}}. \end{align*} If $ n> 3q_n$, then $\frac{(n-2q_n)}{(2n-3q_n)}>\frac{1}{3} $ and $B(3,n,q_n)/B(3,n-1,q_n)=1-\frac{q_n(n-2q_n)}{(2n-3q_n)n}\leq 1-\frac{q_n}{3n}$. Thus, if $ n> 3q_n$ and $\ q_n>2$, then\begin{align*} B(3,n,q_n)&= B(3,3q_n,q_n)\prod_{j=3q_n+1}^nB(3,j,q_n)/B(3,j-1,q_n)\\ &\leq B(3,2q_n,q_n)\prod_{j=3q_n+1}^n\left(1-\frac{q_n}{3j}\right)\leq (q_n/2+1)^{-\frac{1}{2}}\prod_{j=3q_n+1}^n\left(1-\frac{1}{j}\right)\\&=(q_n/2+1)^{-\frac{1}{2}}\frac{3q_n}{n}\leq cn^{-\frac{1}{2}}.\end{align*}If $ n> 3q_n$ and $q=2$, then\begin{align*} B(3,n,q_n)= {2\choose 1}{n-2\choose 1}/{n\choose 2}=\frac{4(n-2)}{n(n-1)}<\frac{4}{n}\leq cn^{-\frac{1}{2}}. \end{align*}Therefore, $B(3,n,q_n)\leq cn^{-\frac{1}{2}} $ is always true and $ |B_m|\leq |B_m^*|=B(3,n,q_n)|I_n|^{m-1}\leq c|I_n|^{m-1}n^{-\frac{1}{2}}.$

{\bf Case 2:} $m=4.$ We have $ B_m^*=\{({R_1},R_2,  {R_{3}},  {R_{4}})\in I_n^{m}|\Psi_{R_1}\Psi_{R_2}=\pm \Psi_{R_3}\Psi_{R_4}\}=\{({R_1},R_2,  {R_{3}},  {R_{4}})\in I_n^{m}|\Psi_{R_1\triangle R_2}=\pm \Psi_{R_3\triangle R_4}\}=\{({R_1},R_2,  {R_{3}},  {R_{4}})\in I_n^{m}|R_1\triangle R_2={R_3\triangle R_4}\}$. If $({R_1},R_2,  {R_{3}},  {R_{4}})\in B_m^*,$ let $A=R_1\triangle R_2={R_3\triangle R_4},$ then $|A|=|R_1\triangle R_2|=|R_1|+| R_2|-2|R_1\cap R_2|=2q_n-2|R_1\cap R_2|$ is even and $|A|\leq 2q_n,\ |A\cap R_1|=|R_1|-|R_1\cap R_2|=q_n-(2q_n-|A|)/2=|A|/2,$ we also have $|A|=|R_1\triangle R_2|=2|R_1\cup R_2|-|R_1|-| R_2|\leq 2n-2q_n.$ For fixed $A$, assume $|A|=2k$, then there are ${2k\choose k}{n-2k\choose q_n-k}$ choices of $R_1$ satisfying $|R_1|=q_n,\ |A\cap R_1|=k$ and $R_2$ is uniquely determined by $R_1,A$. Similarly there are ${2k\choose k}{n-2k\choose q_n-k}$ choices of $(R_3,R_4).$ Moreover for every fixed integer $k,\ 0\leq k\leq \min(q_n,n-q_n),$ there are ${n\choose 2k}$ choices of $A$ satisfying $|A|=2k$. Therefore, we have  $$|B_m^*|=\sum\limits_{k=0}^{\min(q_n,n-q_n)}{n\choose 2k}{2k\choose k}^2{n-2k\choose q_n-k}^2.$$ Notice that \begin{align*} {n\choose 2k}{2k\choose k}{n-2k\choose q_n-k}&=\frac{n!}{(2k)!(n-2k)!}\frac{(2k)!}{(k!)^2}\frac{(n-2k)!}{(q_n-k)!(n-q_n-k)!}\\&=\frac{n!}{k!(q_n-k)!k!(n-q_n-k)!}
={n\choose q_n}{q_n\choose k}{n-q_n\choose k}, \end{align*} then we have $$|B_m^*|=\sum\limits_{k=0}^{\min(q_n,n-q_n)}{n\choose q_n}^2{q_n\choose k}^2{n-q_n\choose k}^2{n\choose 2k}^{-1}$$ and $$B(4,n,q_n)=|B_m^*|/|I_n|^{3}={n\choose q_n}^{-1}\sum\limits_{k=0}^{\min(q_n,n-q_n)}{q_n\choose k}^2{n-q_n\choose k}^2{n\choose 2k}^{-1}.$$ Therefore, $B(4,n,q_n)=B(4,n,n-q_n)$, thus we only need to consider the case $2\leq q_n\leq n/2.$ If $2\leq q_n\leq n/10,$ then ${n\choose 2k}={n\choose k}{n-k\choose k}/{2k\choose k}\geq {n-q_n\choose k}^2/{2k\choose k}$ for $0\leq k\leq q_n$ and we have
\begin{align*} B(4,n,q_n)&={n\choose q_n}^{-1}\sum\limits_{k=0}^{q_n}{q_n\choose k}^2{n-q_n\choose k}^2{n\choose 2k}^{-1}\leq {n\choose q_n}^{-1}\sum\limits_{k=0}^{q_n}{2k\choose k}{q_n\choose k}^2\\ &\leq {n\choose q_n}^{-1}\sum\limits_{k=0}^{q_n}2^{2k}{2q_n\choose 2k}\leq {n\choose q_n}^{-1}\sum\limits_{j=0}^{2q_n}2^{j}{2q_n\choose j}={n\choose q}^{-1}3^{2q_n}\\&=\prod_{j=0}^{q_n-1}\frac{1+j}{n-j}3^{2q_n}\leq \frac{1}{n}\left(\frac{q_n}{n-q_n}\right)^{q_n-1}9^{q_n}\leq \frac{9}{n}\leq cn^{-\frac{1}{2}}. \end{align*} If $n/10\leq q_n\leq n/2,$ then for $0\leq k\leq q_n$ and $n$ even,we have \begin{align*} &{q_n\choose k}{n-q_n\choose k}{n\choose 2k}^{-1}={2k\choose k}\prod_{j=0}^{k-1}\frac{(q_n-j)(n-q_n-j)}{(n-2j)(n-2j-1)}\\ &\leq {2k\choose k}\prod_{j=0}^{k-1}\frac{(n/2-j)^2}{(n-2j)(n-2j-1)}=\prod_{j=1}^{k}\frac{2(2j-1)}{j}\prod_{j=0}^{k-1}\frac{(n/2-j)}{2(n-2j-1)}\\
&=\prod_{j=1}^{k'}\left(\frac{2j-1}{2j}\frac{n+2-2j}{n+1-2j}\right)\leq \prod_{j=1}^{k'}\left(\frac{2j-1}{2j+1}\frac{n+2-2j}{n-2j}\right)^{\frac{1}{2}}\\
&=\left(\frac{n}{(2k'+1)(n-2k')}\right)^{\frac{1}{2}}\leq \left(\frac{2}{2k'+1}\right)^{\frac{1}{2}} \end{align*} where $k'=\min(k,n/2-k)\leq n/4.$ For $j=0,1,2$, we have\begin{align*} a_j:=\sum\limits_{k=0}^{q_n}{k\choose j}{q_n\choose k}{n-q_n\choose k}= \sum\limits_{k=j}^{q_n}{q_n\choose j}{q_n-j\choose k-j}{n-q_n\choose n-q_n-k}={q_n\choose j}{n-j\choose n-q_n-j} \end{align*} and $(k-\mu)^2=2{k\choose 2}-(2\mu-1){k\choose 1}+\mu^2{k\choose 0}$. Therefore,  for $\mu=\frac{q(n-q)}{n}$, we have\begin{align*} &\sum\limits_{k=0}^{q_n}(k-\mu)^2{q_n\choose k}{n-q_n\choose k}=2a_2-(2\mu-1)a_1+\mu^2a_0\\&= 2{q_n\choose 2}{n-2\choose n-q_n-2}-(2\mu-1)q_n{n-1\choose n-q_n-1}+\mu^2{n\choose n-q_n}\\&={n\choose q_n}\left(\frac{q_n(q_n-1)(n-q_n)(n-q_n-1)}{n(n-1)}-(2\mu-1)\frac{q_n(n-q_n)}{n}+\mu^2\right)\\&={n\choose q_n}\mu\left(\frac{(q_n-1)(n-q_n-1)}{(n-1)}-\mu+1\right)={n\choose q_n}\mu\frac{q_n(n-q_n)}{n(n-1)}\end{align*} as $9n/100\leq \mu\leq n/4$ (using $n/10\leq q_n\leq n/2$) and $ \frac{(k-\mu)^2}{\mu^2}+\mu^{-\frac{1}{2}}\geq c^{-1}\left(\frac{2}{2k'+1}\right)^{\frac{1}{2}} $ for $k'=\min(k,n/2-k)$ (one can check this by discussing the cases $|k-\mu|\leq \mu/2$ and $|k-\mu|\geq \mu/2$ separately), which implies\begin{align*} &{q_n\choose k}{n-q_n\choose k}{n\choose 2k}^{-1}\leq c\left(\frac{(k-\mu)^2}{\mu^2}+\mu^{-\frac{1}{2}}\right)\end{align*} and that\begin{align*} B(4,n,q_n)&={n\choose q_n}^{-1}\sum\limits_{k=0}^{q_n}{q_n\choose k}^2{n-q_n\choose k}^2{n\choose 2k}^{-1}\\&\leq c{n\choose q_n}^{-1}\sum\limits_{k=0}^{q_n}{q_n\choose k}{n-q_n\choose k}\left(\frac{(k-\mu)^2}{\mu^2}+\mu^{-\frac{1}{2}}\right)\\ &= c{n\choose q_n}^{-1}{n\choose q_n}\left(\frac{\mu}{\mu^2}\frac{q_n(n-q_n)}{n(n-1)}+\mu^{-\frac{1}{2}}\right)\\&=c\left(\frac{1}{n-1}+\mu^{-\frac{1}{2}}\right)\leq cn^{-\frac{1}{2}}. \end{align*} If $n/2\leq q_n\leq n-2,$ then $B(4,n,q_n)=B(4,n,n-q_n)\leq cn^{-\frac{1}{2}}.$ Therefore, $B(4,n,q_n)\leq cn^{-\frac{1}{2}} $ is always true and $ |B_m|\leq |B_m^*|=B(4,n,q_n)|I_n|^{m-1}\leq c|I_n|^{m-1}n^{-\frac{1}{2}}.$

{\bf Case 3:} $m\geq4.$ We need to prove that $ |B_m^*|\leq |B_4^*||I_n|^{m-4}.$ Since $\Tr \Psi_{R_1}\cdots \Psi_{R_m}\\=0$ for
$\Psi_{R_1}\cdots \Psi_{R_m}\neq\pm I$, we have
\begin{align*} |B_m^*|&=L_n^{-2}\sum\limits_{({R_1},...,  {R_{m}})\in I_n^{m}}(\Tr \Psi_{R_1}\cdots \Psi_{R_m})^2\\
&= L_n^{-2}\sum\limits_{({R_1},...,  {R_{m}})\in I_n^{m}}\Tr (\Psi_{R_1}\cdots \Psi_{R_m}\otimes\Psi_{R_1}\cdots \Psi_{R_m})\\
&= L_n^{-2}\sum\limits_{({R_1},...,  {R_{m}})\in I_n^{m}}\Tr (\Psi_{R_1}\otimes\Psi_{R_1})\cdots (\Psi_{R_m}\otimes \Psi_{R_m})\\
&=L_n^{-2}\Tr \left(\sum\limits_{{R_1}\in I_n}\Psi_{R_1}\otimes\Psi_{R_1}\right)\cdots
\left(\sum\limits_{{R_m}\in I_n}\Psi_{R_m}\otimes \Psi_{R_m}\right)\\
&=L_n^{-2}\Tr \left(\sum\limits_{{R}\in I_n}\Psi_{R}\otimes\Psi_{R}\right)^m. \end{align*}
Since $\Psi_{R} $ is a $L_n\times L_n$ Hermitian or anti-Hermitian matrix, $\Psi_{R}\otimes\Psi_{R} $ is a $L_n^2\times L_n^2$
Hermitian matrix and $\widetilde{H}=\sum\limits_{{R}\in I_n}\Psi_{R}\otimes\Psi_{R}$ is a $L_n^2\times L_n^2$ Hermitian matrix.
Let's assume that the eigenvalues of $\widetilde{H}$ are $\mu_j\in\R\ (1\leq j\leq L_n^2)$,  then we have
\begin{align*} |B_m^*|&=L_n^{-2}\Tr \widetilde{H}^m=L_n^{-2}\sum_{j=1}^{L_n^2}\mu_j^m. \end{align*}
As $\Psi_{R}^2=\pm I $, we have $(\Psi_{R}\otimes\Psi_{R})^2=I $ and $|\Psi_{R}\otimes\Psi_{R} x|=|x| $ for $x\in \mathbb{C}^{L_n^2}$.
Therefore, $|\widetilde{H} x|\leq \sum\limits_{{R}\in I_n}|\Psi_{R}\otimes\Psi_{R}x|=|I_n||x|$ which  implies $|\mu_j|\leq |I_n|. $

 Now we have \begin{align*} |B_m|&\leq|B_m^*|=L_n^{-2}\sum_{j=1}^{L_n^2}\mu_j^m\leq L_n^{-2}\sum_{j=1}^{L_n^2}\mu_j^4|I_n|^{m-4}
=|B_4^*||I_n|^{m-4}\\&=B(4,n,q_n)|I_n|^{m-1}\leq c|I_n|^{m-1}n^{-\frac{1}{2}}. \end{align*} This completes the proof.
\end{proof}

%\begin{rem}\label{oood}
%	If $1\leq q_n\leq n/2$ is odd, following the proofs of Lemmas \ref{lemma11} and \ref{lemma10},  Lemma \ref{lemma11} is true with  ${m}^a_k$ replaced by $\widetilde{m}^a_k$ which is defined in \eqref{malk}.
%\end{rem}
\subsection{Proof of Theorem \ref{main3}}
%For the limit of ${n\choose q_n} \text{var}[L_n^{-1} \Tr H^k],$ we need to improve the estimate of $|B_m|. $

Let $f(x)=\sum\limits_{k=0}^ma_kx^k$ be a real polynomial,  then   $$\mathcal L_n(f)=\sum\limits_{k=0}^ma_k\langle x^k,   \rho_n \rangle=\sum\limits_{k=0}^ma_kL_n^{-1} \Tr H^k. $$
We first have\begin{lem}\label{lemma12}With the same assumptions as in Theorem \ref{main3}, if $2\leq q_n\leq  n /2$ is even,  then $$  \lim\limits_{n\to+\infty}{n\choose q_n} \text{var}[\mathcal L_n(f)]= \langle xf'/2,  \rho_{\infty} \rangle^2(\gamma-1).$$
\end{lem}\begin{proof} Since $\mathcal L_n(f)-a_0=\sum\limits_{k=1}^ma_kL_n^{-1} \Tr H^k, $ by Lemma \ref{lemma11}, we have\begin{align*}&\lim\limits_{n\to+\infty}{n\choose q_n}\text{var}[\mathcal L_n(f)]\\&=\lim\limits_{n\to+\infty}{n\choose q_n}\sum\limits_{k=1}^m\sum\limits_{k'=1}^ma_ka_{k'}\text{cov}(L_n^{-1} \Tr H^k,L_n^{-1} \Tr H^{k'})\\&=\sum\limits_{k=1}^m\sum\limits_{k'=1}^ma_ka_{k'}(m_k^ak/2)(m_{k'}^ak'/2)(\gamma-1)\\&
=\left(\sum\limits_{k=1}^ma_km_k^ak/2\right)^2(\gamma-1) =\left(\sum\limits_{k=1}^ma_k\langle x^k,  \rho_{\infty} \rangle k/2\right)^2(\gamma-1)\\& =\left\langle \sum\limits_{k=1}^ma_kx^kk/2,  \rho_{\infty} \right\rangle^2(\gamma-1).
\end{align*} Since $xf'(x)/2=\sum\limits_{k=1}^ma_kx^kk/2,$ this completes the proof.\end{proof}
Now we can finish the proof of Theorem \ref{main3}.
%\begin{lem}\label{lemma13}If $2\leq q_n\leq n-2,\  q_n^2/n\to a\in[0,+\infty],\ \e J_R^4=\gamma,$ $f(x)=\sum\limits_{k=0}^ma_kx^k$ is a
%polynomial, $a_k\in\R,$ then in the sense of distribution,$$  {n\choose q_n}^{\frac{1}{2}} (\langle f, \tilde \rho_n(\lambda)\rangle-
%\e[\langle f, \tilde \rho_n(\lambda)\rangle])\Rightarrow \langle xf'/2,  \rho_{\infty}(\lambda)\rangle J,$$ where $J$ has normal
%distribution of mean 0 and variance $ \gamma-1.$
%\end{lem}

\begin{proof}Let $2\leq q_n\leq n/2$ be even. We first consider the case $f(x)=x^2$ where $$\langle f,   \rho_n(\lambda)\rangle=L_n^{-1} \Tr H^2=\dfrac 1{L_n} \dfrac {i^{q_n}}{{n\choose q_n}} \sum\limits_{{R_1},{R_2}\in I_n} J_{R_1} J_{R_2}  \Tr \Psi_{R_1} \Psi_{R_2}. $$ As discussed in \S \ref{nota}, if $R_1\neq R_2$, then $\Psi_{R_1} \Psi_{R_2}\neq \pm I,\ \Tr \Psi_{R_1} \Psi_{R_2}=0; $ if $R_1= R_2$, then $i^{q_n}\Psi_{R_1} \Psi_{R_2}= I.$ Therefore $$\langle f,   \rho_n(\lambda)\rangle=L_n^{-1} \Tr H^2= \dfrac {1}{{n\choose q_n}} \sum\limits_{R\in I_n} J_{R}^2. $$ Since $\e[J_{R}^2]=1,$ we have $$\langle f,  \rho_n(\lambda)\rangle-\e[\langle f,   \rho_n(\lambda)\rangle]= \dfrac {1}{{n\choose q_n}} \sum\limits_{R\in I_n} (J_{R}^2-1). $$ The random variables $J_{R}^2-1 $ are independent with $\mathbb E[J_{R}^2-1]=0$ and  $\text{var}[J_{R}^2-1]=\e[J_R^4]-1=\gamma-1.$ By   assumptions, $\mathbb E[(J_{R}^2-1)^4] $ is uniformly bounded, therefore, the random variables $J_{R}^2-1,\ R\in I_n $ satisfy the Lyapunov condition. Thus,  by Lindeberg-Feller central limit law,  we have $$ {n\choose q_n}^{\frac{1}{2}} (\langle f,   \rho_n\rangle-\e[\langle f,  \rho_n\rangle])\Rightarrow J,$$
  where $J$ is Gaussian random variable with mean 0 and variance $ \gamma-1$.
  Since $m_2^a=1$ for $a\in [0,+\infty]$, thus $ \langle xf'/2,  \rho_{\infty}\rangle=\langle x^2,  \rho_{\infty}\rangle=m_2^a=1,$ this will imply that Theorem \ref{main3} is true for $f(x)=x^2$.

Now we consider the case for general polynomials $f(x)$. Let $\mu=\langle xf'/2,  \rho_{\infty}\rangle$ and define $f_1=f-\mu x^2,$ then we have $ \langle xf_1'/2,  \rho_{\infty}\rangle=\langle xf'/2,  \rho_{\infty}\rangle-\mu\langle x^2,  \rho_{\infty}\rangle=0.$   Thus, by Lemma \ref{lemma12}, we have\begin{align*}\lim\limits_{n\to+\infty}{n\choose q_n}\text{var}[\langle f_1,   \rho_n\rangle]=\left\langle xf_1'/2,  \rho_{\infty}\right\rangle^2(\gamma-1)=0.
\end{align*} Therefore,  we have $${n\choose q_n}^{\frac{1}{2}} (\langle f_1,  \rho_n\rangle-\e[\langle f_1,  \rho_n\rangle])\to 0 $$ in probability. Now  we have \begin{align*}&{n\choose q_n}^{\frac{1}{2}} (\langle f,   \rho_n\rangle-\e[\langle f,   \rho_n\rangle])\\&={n\choose q_n}^{\frac{1}{2}} (\langle f_1,   \rho_n\rangle-\e[\langle f_1, \rho_n\rangle])+ \mu{n\choose q_n}^{\frac{1}{2}} (\langle x^2,   \rho_n\rangle-\e[\langle x^2, \rho_n\rangle]).\end{align*}  The first term tends to 0 in probability and the second term tends to $\mu J$ in distribution, therefore, we prove
%  $\\  {n\choose q_n}^{\frac{1}{2}} (\langle x^2,  \rho_n(\lambda)\rangle-\e[\langle x^2,   \rho_n(\lambda)\rangle])\Rightarrow J,$  and

$${n\choose q_n}^{\frac{1}{2}} (\langle f,  \rho_n\rangle-\e[\langle f,  \rho_n\rangle])\Rightarrow \mu J.$$ By definition of $\mu$, Theorem \ref{main3} is proved when $2\leq q_n\leq n/2$ is even.   
For even $n/2\leq q_n< n$, the results follow immediately since there is a natural symmetry between the cases of $n-q_n$ and $q_n$. For odd  $q_n$, the results follow with almost the same proof.   Thus we complete the proof of Theorem \ref{main3}.\end{proof}

 \begin{rem}\label{redd}As a remark, let $c_k$ be the constants in \eqref{upp}, then   Theorem \ref{main3} holds for a class of analytic functions $f(x)=\sum\limits_{k=0}^{\infty}a_kx^k$ with $ \sum\limits_{k=0}^{\infty}|a_k|c_k^{\frac{1}{2}}<+\infty.$\end{rem} To see this,  we first have \begin{equation}\label{keyy} {n\choose q_n}\text{var}[\langle f,   \rho_n\rangle]={n\choose q_n}\sum\limits_{k=1}^{+\infty}\sum\limits_{k'=1}^{+\infty}a_ka_{k'}\text{cov}(L_n^{-1} \Tr H^k,L_n^{-1} \Tr H^{k'}).
\end{equation}By \eqref{upp},  we have the upper bound $${n\choose q_n}|\text{cov}(L_n^{-1} \Tr H^k,L_n^{-1} \Tr H^{k'})|\leq (c_kc_{k'})^{\frac{1}{2}}. $$
  Then by assumption, we have\begin{align*}\sum\limits_{k=1}^{+\infty}\sum\limits_{k'=1}^{+\infty}|a_ka_{k'}|(c_kc_{k'})^{\frac{1}{2}}=
\left(\sum\limits_{k=0}^{\infty}|a_k|c_k^{\frac{1}{2}}\right)^2<+\infty .
\end{align*}  Therefore,  by the dominated convergence theorem,  Lemma \ref{lemma12} holds for $f(x)$ if we take the limit on both sides of  \eqref{keyy}, and hence Theorem \ref{main3}.

\section{Improved CLT for the Gaussian SYK}\label{cltt}
In this section, we will prove that  the CLT for the linear statistic of the Gaussian SYK model holds for a more general class of functions. We will prove Theorem \ref{main3a} by  approximations making use of the 
F\'ejer kernel and Theorem \ref{main3}.    
\subsection{Estimate of variance}
 We first  need  the following estimate in the Gaussian case which is more precise compared with Lemma \ref{lemmma}.
\begin{lem}\label{lemma6a}Let $q_n$ be even,  for the Gaussian SYK model,   we have $$   \text{var}[L_n^{-1} \Tr H^k]\leq c_k{n\choose q_n}^{-1}$$ for any $ k\geq1$ with $c_k=2^kk!k^2.$
\end{lem}\begin{proof} %Since \begin{equation}\label{mnk1} \frac 1{L_n}    \Tr H^k= \frac 1{L_n} \frac {i^{q_nk/2}}{{n\choose q_n}^{k/2}} \sum_{{R_1},...,  {R_k}\in I_n} J_{R_1}\cdots J_{R_k}  \Tr \Psi_{R_1}\cdots \Psi_{R_k}, \end{equation} we have \begin{align*} \text{var}[L_n^{-1} \Tr H^k]= \frac 1{L_n^2} \frac {(-1)^{q_nk/2}}{{n\choose q_n}^{k}} \sum_{{R_1},...,  {R_{2k}}\in I_n} \text{cov}(J_{R_1}\cdots J_{R_k},J_{R_{k+1}}\cdots J_{R_{2k}})\cdot\\  \Tr \Psi_{R_1}\cdots \Psi_{R_k}\Tr \Psi_{R_{k+1}}\cdots \Psi_{R_{2k}}. \end{align*}For every ${R_1},...,  {R_{2k}}\in I_n $ and $A\in I_n$. Let $ \#A=|\{j|1\leq j\leq 2k,R_j=A\}|.$

%If some $R_i$ only appears once in $(R_1,\cdots, R_{2k})$, then $\e J_{R_1}\cdots J_{R_{2k}} =0$ and $\e J_{R_1}\cdots J_{R_{k}} =0 $ if $1\leq i\leq k,$ $\e J_{R_{k+1}}\cdots J_{R_{2k}} =0 $ if $k+1\leq i\leq 2k,$ since $J_R$ are standard Gaussian, we have $\text{cov}(J_{R_1}\cdots J_{R_k},J_{R_{k+1}}\cdots J_{R_{2k}})=\e J_{R_1}\cdots J_{R_{2k}}-\e J_{R_1}\cdots J_{R_{k}}\e J_{R_{k+1}}\cdots J_{R_{2k}} =0 $.
As in the beginning of the proof of Lemma \ref{lemma11},  we first have \begin{align*} \text{var}[L_n^{-1} \Tr H^k]&= \frac 1{L_n^2} \frac {(-1)^{q_nk/2}}{{n\choose q_n}^{k}} \left(\sum_{({R_1},...,  {R_{2k}})\in P_2(I_n^{2k})}+ \sum_{({R_1},...,  {R_{2k}})\in I_n^{2k}\setminus P_2(I_n^{2k}), \# R_i\geq 2}\right)\\& \text{cov}(J_{R_1}\cdots J_{R_k},J_{R_{k+1}}\cdots J_{R_{2k}})\cdot  \Tr \Psi_{R_1}\cdots \Psi_{R_k}\Tr \Psi_{R_{k+1}}\cdots \Psi_{R_{2k}}\\ &:=V_1+V_2. \end{align*}We can write $J_{R_1}\cdots J_{R_k}=J_{R_1'}^{a_1}\cdots J_{R_l'}^{a_l},\ J_{R_{k+1}}\cdots J_{R_{2k}}=J_{R_1'}^{b_1}\cdots J_{R_l'}^{b_l} $ such that $J_{R_1'},\cdots, J_{R_l'} \in I_n$ are distinct and $a_j,b_j\geq 0$ are integers. Then we have $J_{R_1}\cdots J_{R_{2k}}\\=J_{R_1'}^{a_1+b_1}\cdots J_{R_l'}^{a_l+b_l}.$  Let $ \gamma_j:=\e J_R^j,\ R\in I_n,\ j\geq 0,\ j\in \mathbb{Z}$, then $ \gamma_j$ satisfies $ \gamma_j=(2j-1)!!$ for $j$ even and $ \gamma_j=0$ for $j$ odd,  $ \gamma_j\geq 0$ and $ \gamma_{j+m}\geq \gamma_j\gamma_m.$ 
Thus we have \begin{align*} & \text{cov}(J_{R_1}\cdots J_{R_k},J_{R_{k+1}}\cdots J_{R_{2k}})\\&=\e J_{R_1}\cdots J_{R_{2k}}-\e J_{R_1}\cdots J_{R_{k}}\e J_{R_{k+1}}\cdots J_{R_{2k}}\\&=\e J_{R_1'}^{a_1+b_1}\cdots J_{R_l'}^{a_l+b_l}-\e J_{R_1'}^{a_1}\cdots J_{R_l'}^{a_l}\e J_{R_1'}^{b_1}\cdots J_{R_l'}^{b_l}\\&=\prod_{j=1}^l\gamma_{a_j+b_j}-\prod_{j=1}^l(\gamma_{a_j}\gamma_{b_j})\geq 0 \end{align*}   and\begin{align*}& \text{cov}(J_{R_1}\cdots J_{R_k},J_{R_{k+1}}\cdots J_{R_{2k}})\\& \leq\prod_{j=1}^l\gamma_{a_j+b_j}=\e J_{R_1'}^{a_1+b_1}\cdots J_{R_l'}^{a_l+b_l}=\e J_{R_1}\cdots J_{R_{2k}}.  \end{align*}By \eqref{tr} where $| L_n^{-1} \Tr \Psi_{R_1}\cdots \Psi_{R_k}|\leq 1$ and $ | L_n^{-1} \Tr \Psi_{R_{k+1}}\cdots \Psi_{R_{2k}}|\leq 1,$ we have\begin{align*} |V_2|\leq&{n\choose q_n}^{-k}\sum_{({R_1},...,  {R_{2k}})\in I_n^{2k}\setminus P_2(I_n^{2k})}| \text{cov}(J_{R_1}\cdots J_{R_k},J_{R_{k+1}}\cdots J_{R_{2k}})|\\ \leq&{n\choose q_n}^{-k}\sum_{({R_1},...,  {R_{2k}})\in I_n^{2k}\setminus P_2(I_n^{2k})}\e J_{R_1}\cdots  J_{R_{2k}}\\ =&{n\choose q_n}^{-k}\left(\sum_{({R_1},...,  {R_{2k}})\in I_n^{2k}}\e J_{R_1}\cdots  J_{R_{2k}}-\sum_{({R_1},...,  {R_{2k}})\in  P_2(I_n^{2k})}\e J_{R_1}\cdots  J_{R_{2k}}\right)\\ =&|I_n|^{-k}\left(\e \left(\sum_{R\in I_n}J_{R}\right)^{2k}-|P_2(I_n^{2k})|\right)=\gamma_{2k}-|I_n|^{-k}|P_2(I_n^{2k})|\\=&(2k-1)!!-|I_n|^{-k}{|I_n|\choose k}k!(2k-1)!!=(2k-1)!!\left(1-\prod_{j=0}^{k-1}\frac{|I_n|-j}{|I_n|}\right)\\ \leq &(2k-1)!!\sum_{j=0}^{k-1}\frac{j}{|I_n|}=(2k-1)!!\frac{k(k-1)}{2|I_n|}. \end{align*}Here, we used the fact that $\frac{1}{\sqrt{|I_n|}}\sum\limits_{R\in I_n} J_{R} $ has the standard Gaussian distribution. For the estimate of $V_1$,  we first easily have $$V_1\leq {n\choose q_n}^{-k}{|P_{2,1}(I_n^{2k})|}.$$ In the Gaussian case, following the proof of Lemma \ref{lemmma}  (see  \cite{FTD1} for the detailed proof), we further have the uniform estimate\begin{align*} |P_{2,1}(I_n^{2k})|&=\sum_{0<m\leq k,2|k-m}{|I_n|-m\choose k-m}{k-m\choose \frac{k-m}{2}}(k!)^2/(2^{{k-m}}m!)\cdot|B_m|\\&\leq\sum_{0<m\leq k,2|k-m}\frac{|I_n|^{k-m}}{(k-m)!}2^{{k-m}}(k!)^2/(2^{{k-m}}m!)\cdot|I_n|^{m-1}\\&=|I_n|^{k-1}\sum_{0<m\leq k,2|k-m}(k!){k\choose m}\leq(2^kk!)|I_n|^{k-1}. \end{align*} Using $|I_n|={n\choose q_n}$, we have \begin{align*} &\text{var}[L_n^{-1} \Tr H^k]=V_1+V_2\leq {n\choose q_n}^{-k}{|P_{2,1}(I_n^{2k})|} +V_2\\ \leq& (2^kk!){n\choose q_n}^{-k}|I_n|^{k-1}+(2k-1)!!\frac{k(k-1)}{2|I_n|}= c_k'{n\choose q_n}^{-1},\end{align*}with $c_k'=2^kk!+(2k-1)!!{k(k-1)}/2\leq 2^kk!+(2k)!!{k(k-1)}/2=(2k)!!(1+{k(k-1)}/2)\leq (2k)!!k^2=2^kk!k^2=:c_k.$ This completes the proof.\end{proof}%As a remark, Lemma \ref{lemma6a} is still true if $q_n$ is odd. 
%If $J_{i_1\cdots i_{q_n}}$ are independent, standard Gaussian random variables, by Lemma \ref{lemma6a} we know that Theorem \ref{main3} is still true for a class of analytic functions $f(x)=\sum\limits_{k=0}^{\infty}a_kx^k$ such that $ \sum\limits_{k=0}^{\infty}|a_k|(2^kk!)^{\frac{1}{2}}k<+\infty.$
\subsection{F\'ejer kernel and approximations}
% Given a bounded Lipschitz function $f$, we further assume that $f'$ is bounded uniformly continuous, we will prove that $f_{\la}:=f*K_{\la}$
% approximates $f$ in $L^\infty(\mathbb R)$,  where 
Let $K_{\la}(x)$ be the F\'ejer kernel\begin{align*} &K_{\la}(x)=\frac{1}{2\pi}\int_{-\la}^{\la}\left(1-\frac{|\xi|}{\la}\right)e^{i\xi x}d\xi=\frac{\la}{2\pi}\left(\frac{\sin(\la x/2)}{\la x/2}\right)^2,\,\,\, \la>0.\end{align*}

 \begin{lem}\label{lemKl} $K_{\la}\in C^{\infty}(\R)$ and the detivatives $|K_{\la}^{(k)}(x)|\leq \dfrac{\la^{k+1}}{2\pi}\max(1,\la x/3)^{-2}. $
\end{lem}\begin{proof}By definition of $K_{\la}(x)$, we have $K_{\la}\in C^{\infty}(\R)$ and $$K_{\la}^{(k)}(x)=\frac{1}{2\pi}\int_{-\la}^{\la}\left(1-\frac{|\xi|}{\la}\right)(i\xi)^ke^{i\xi x}d\xi$$ for $k\geq 0,k\in\mathbb{Z}^+.$ Therefore,
\begin{align*} &|K_{\la}^{(k)}(x)|\leq \frac{1}{2\pi}\int_{-\la}^{\la}\left(1-\frac{|\xi|}{\la}\right)|\xi|^kd\xi\leq \frac{\la^k}{2\pi}\int_{-\la}^{\la}\left(1-\frac{|\xi|}{\la}\right)d\xi=\dfrac{\la^{k+1}}{2\pi}.\end{align*}On the other hand, if $x\in\R\setminus\{0\}$,  
%\begin{align*} &|K_{\la}^{}(x)|=\frac{\la}{2\pi}\left(\frac{\sin(\la x/2)}{\la x/2}\right)^2\leq \frac{\la}{2\pi}\left(\frac{1}{\la x/2}\right)^2\leq \frac{\la}{2\pi}(\la x/3)^{-2},\end{align*} 
we integrate by parts to get \begin{align*} K_{\la}^{(k)}(x)=&\frac{1}{-2\pi ix}\int_{-\la}^{\la}i^k\left(k\xi^{k-1}-(k+1)\frac{|\xi|}{\la}\xi^{k-1}\right)e^{i\xi x}d\xi\\=&\frac{ i^k}{2\pi( ix)^2}\left(\int_{-\la}^{\la}\left(k(k-1)\xi^{k-2}-k(k+1)\frac{|\xi|}{\la}\xi^{k-2}\right)e^{i\xi x}d\xi+\xi^{k-1}e^{i\xi x}\Big|_{-\la}^{\la}\right),\end{align*}therefore,\begin{align*}  |K_{\la}^{(k)}(x)|\leq&\frac{ 1}{2\pi x^2}\left(\int_{-\la}^{\la}\left|k(k-1)\xi^{k-2}-k(k+1)\frac{|\xi|}{\la}\xi^{k-2}\right|d\xi+2\la^{k-1}\right)\\=&\frac{ 1}{2\pi x^2}\left(\int_{-\la}^{\la}\left|k(k-1)\left(1-\frac{|\xi|}{\la}\right)\xi^{k-2}-2k\frac{|\xi|}{\la}\xi^{k-2}\right|d\xi
+2\la^{k-1}\right)\\ \leq&\frac{ 1}{2\pi x^2}\left(\int_{-\la}^{\la}\left(k(k-1)\left(1-\frac{|\xi|}{\la}\right)|\xi|^{k-2}+2k\frac{|\xi|^{k-1}}{\la}\right)d\xi
+2\la^{k-1}\right)\\=&\frac{ 1}{2\pi x^2}\left(\int_{-\la}^{\la}\left(k(k-1)|\xi|^{k-2}-k(k-3)\frac{|\xi|^{k-1}}{\la}\right)d\xi
+2\la^{k-1}\right)\\=&\frac{ 1}{2\pi x^2}\left(2k\la^{k-1}\chi_{\{k>1\}}-2(k-3)\frac{\la^{k}}{\la}
+2\la^{k-1}\right)\\ \leq&\frac{ 1}{2\pi x^2}\left(2k\la^{k-1}-2(k-3)\la^{k-1}
+2\la^{k-1}\right)=\frac{ 8\la^{k-1}}{2\pi x^2}\leq \dfrac{\la^{k+1}}{2\pi}(\la x/3)^{-2}.\end{align*}Thus we have $|K_{\la}^{(k)}(x)|\leq \dfrac{\la^{k+1}}{2\pi} $ for $x\in \R$ and $ |K_{\la}^{(k)}(x)|\leq \dfrac{\la^{k+1}}{2\pi}(\la x/3)^{-2} $ for $x\in\R\setminus\{0\}.$ This completes the proof.\end{proof}

\begin{lem}\label{lemfl}For $ f\in L^{\infty}(\R)$, we have $f_{\la}:=f*K_{\la}\in C^{\infty}(\R)$ and $|f_{\la}^{(k)}(x)|\leq 2{\la^{k}}\|f\|_{L^{\infty}(\R)},$ thus $f_{\la}$ is real analytic.
\end{lem}\begin{proof}By Lemma \ref{lemKl}, we have $K_{\la}^{(k)}\in L^1(\R), $ actually\begin{align*}\|K_{\la}^{(k)}\|_{L^{1}(\R)}&\leq \dfrac{\la^{k+1}}{2\pi}\|\max(1,\la x/3)^{-2}\|_{L^{1}(\R)}=\dfrac{\la^{k+1}}{2\pi}\dfrac{3}{\la} \|\max(1, x)^{-2}\|_{L^{1}(\R)}\\&=\dfrac{\la^{k+1}}{2\pi}\dfrac{3}{\la}\cdot4 =\dfrac{12\la^{k}}{2\pi}\leq 2\la^k.\end{align*} 
Since $f_{\la}=f*K_{\la},\ f\in L^{\infty}(\R)$ and $K_{\la}\in C^{\infty}(\R)$, thus $f_{\la}\in C^{\infty}(\R)$ and $f_{\la}^{(k)}=f*K_{\la}^{(k)}.$ Furthermore, 
$$|f_{\la}^{(k)}(x)|\leq \|f\|_{L^{\infty}(\R)}\|K_{\la}^{(k)}\|_{L^{1}(\R)}\leq 2{\la^{k}}\|f\|_{L^{\infty}(\R)}.$$
%Thus $|f_{\la}^{(k)}(x)|\leq \|f\|_{L^{\infty}(\R)}\|K_{\la}^{(k)}\|_{L^{1}(\R)}\leq 2{\la^{k}}\|f\|_{L^{\infty}(\R)}.$
 We can assume $f$ is real valued. By Taylor expansion,  for $n\in \mathbb{Z}^+,$ we have\begin{align*} &f_{\la}(x)=\sum_{k=0}^{n-1}\frac{f_{\la}^{(k)}(0)}{k!}x^k+\frac{f_{\la}^{(n)}(\theta x)}{n!}x^n,\ \ x\in\R,\ \end{align*}here $ \theta=\theta(n,x)\in (0,1).$ Now we have $$ \dfrac{|f_{\la}^{(n)}(\theta x)|}{n!}|x|^n\leq \dfrac{2{\la^{n}}\|f\|_{L^{\infty}(\R)}}{n!}|x|^n=\dfrac{2{(\la|x|)^{n}}\|f\|_{L^{\infty}(\R)}}{n!}.$$ Since $ \lim\limits_{n\to+\infty}\dfrac{{(\la|x|)^{n}}}{n!}=0,$ this implies $ \lim\limits_{n\to+\infty}\dfrac{f_{\la}^{(n)}(\theta x)}{n!}x^n=0$ and\begin{align*} &f_{\la}(x)=\sum_{k=0}^{+\infty}\frac{f_{\la}^{(k)}(0)}{k!}x^k,\ \ x\in\R,\ \end{align*}thus $f_{\la}$ is real analytic. This completes the proof.\end{proof}\begin{lem}\label{lemfl1}If $ f\in L^{\infty}(\R)$ is uniformly continuous, then $ \lim\limits_{\la\to+\infty}\|f-f_{\la}\|_{L^{\infty}(\R)}=0$.
\end{lem}\begin{proof}Let $ \omega(a)=\sup\limits_{x\in\R}|f(x)-f(x-a)|$ for $a\in\R,$ then $0\leq \omega(a)\leq 2\|f\|_{L^{\infty}(\R)},\ \omega(a+b)\leq \omega(a)+\omega(b)$ and $\lim\limits_{a\to0}\omega(a)=0$ since $f$ is uniformly continuous, which also implies the continuity of $\omega$. By definition of $f_{\la}$ and the fact that $\int_{\R}K_{\la}(y)dy=1,$we have\begin{align*} &|f(x)-f_{\la}(x)|=\frac{1}{\pi}\left|\int_{\R}\left(f(x)-f\left(x-\frac{2y}{\la}\right)\right)\left(\frac{\sin y}{y}\right)^2 dy\right|\\ \leq &\frac{1}{\pi}\int_{\R}\left|f(x)-f\left(x-\frac{2y}{\la}\right)\right|\left(\frac{\sin y}{y}\right)^2 dy\leq \frac{1}{\pi}\int_{\R}\omega\left(\frac{2y}{\la}\right)\left(\frac{\sin y}{y}\right)^2 dy,\end{align*}thus\begin{align*} &\|f-f_{\la}\|_{L^{\infty}(\R)}\leq \frac{1}{\pi}\int_{\R}\omega\left(\frac{2y}{\la}\right)\left(\frac{\sin y}{y}\right)^2 dy.\end{align*}By dominated convergence theorem,  we  have \begin{align*} &\limsup\limits_{\la\to+\infty}\|f-f_{\la}\|_{L^{\infty}(\R)}\leq \frac{1}{\pi}\lim\limits_{\la\to+\infty}\int_{\R}\omega\left(\frac{2y}{\la}\right)\left(\frac{\sin y}{y}\right)^2dy\\=&\frac{1}{\pi}\int_{\R}\lim\limits_{\la\to+\infty}\omega\left(\frac{2y}{\la}\right)\left(\frac{\sin y}{y}\right)^2 dy=\frac{1}{\pi}\int_{\R}0\left(\frac{\sin y}{y}\right)^2 dy=0,\end{align*}which completes the proof.\end{proof}
\subsection{Proof of Theorem \ref{main3a}}
To prove Theorem \ref{main3a}, we further need two lemmas. 
\begin{lem}\label{lemvar}For the Gaussian SYK model,  let the test function $f:\R\to\R$ be Lipschitz. Then $${n\choose q_n} \text{var}[\mathcal L_n(f)]\leq C\|f'\|_{L^{\infty}(\R)}^2$$
for some universal constant $C$. 
\end{lem}\begin{proof}%Without lose of generality, we can assume $\|f'\|_{L^{\infty}(\R)}=1$.
% For $x=(J_R)_{R\in I_n}\in \R^{{n\choose q_n}},$ define $H_n(x)$ by \begin{align*}H_n(x):=i^{[q_n/2]}{n\choose q_n}^{-1/2}\sum_{R\in I_n}J_R\Psi_R,
%\end{align*}
 The linear statistic $ \mathcal L_n(f) $ is ${n\choose q_n}^{-1/2}\|f'\|_{L^{\infty}(\R)} $-Lipschitz  if we view $ \mathcal L_n(f) $ as a function of Gaussian vectors $x:=(J_R)_{R\in I^n}\in \mathbb R^{{n\choose q_n}}$ (see part (a) of Lemma 1 in \cite{FTD3}). By the standard   concentration of
measure theorem for the Gaussian vectors (see \cite{AGZ}), 
%\begin{prop}\label{prop1}Let $(Z_k)_{1\leq k\leq N}$ be a standard $N$-dimensional Gaussian random vector, and let $F:\R^{N}\to\R $ be Lipschitz with Lipschitz constant $L.$ There are universal constants $C,c>0$ such that for $t>0,$
%\begin{align*}\mathbb{P}[|F(Z_1,\cdots, Z_N)-\e F(Z_1,\cdots, Z_N)|>t]\leq Ce^{-ct^2/L^2}.
%\end{align*}
%\end{prop}By Proposition \ref{prop1} 
we have\begin{align*}\mathbb{P}[|\mathcal L_n(f)-\e\mathcal L_n(f)|>t]\leq Ce^{-ct^2/L^2},\ t>0,
\end{align*}where $L={n\choose q_n}^{-1/2}\|f'\|_{L^{\infty}(\R)}. $ Therefore, we have 
\begin{align*}\text{var}[\mathcal L_n(f)]&=\e|\mathcal L_n(f)-\e\mathcal L_n(f)|^2=\int_0^{+\infty}2t\mathbb{P}[|\mathcal L_n(f)-\e\mathcal L_n(f)|>t]dt\\ &\leq \int_0^{+\infty}2tCe^{-ct^2/L^2}dt\leq CL^{2} =C\|f'\|_{L^{\infty}(\R)}^2{n\choose q_n}^{-1}.
\end{align*}This completes the proof.\end{proof}\begin{lem}\label{lemma12a}For the Gaussian SYK model, let $f=f_1+f_2$ such that $f_1$ is Lipschitz   and $f'_1$ is bounded uniformly continuous, $f_2$ is a polynomial, $f_1,f_2$ are real valued, then $$  \lim\limits_{n\to+\infty}{n\choose q_n} \text{var}[\mathcal L_n(f)]=2 \langle xf'/2,  \rho_{\infty}\rangle^2.$$
\end{lem}\begin{proof}Let $g=f_1',\ g_{\la}=g*K_{\la}$ for $\la>0$, by Lemma \ref{lemfl1}, we have $$ \lim\limits_{\la\to+\infty}\|g-g_{\la}\|_{L^{\infty}(\R)}=0.$$ By Lemma \ref{lemfl}, $ g_{\la}$ is real analytic and $$|g_{\la}^{(k)}(x)|\leq 2{\la^{k}}\|g\|_{L^{\infty}(\R)}=2{\la^{k}}\|f_1'\|_{L^{\infty}(\R)}.$$ As $f_2$ is a polynomial, we can write $f_2(x)=\sum\limits_{k=0}^ma_kx^k.$ Let $F_{\la}(x)=\int_0^xg_{\la}(y)dy+f_2(x),$ by Taylor expansion we have \begin{align*} &F_{\la}(x)=\sum_{k=0}^{+\infty}\frac{g_{\la}^{(k)}(0)}{(k+1)!}x^{k+1}+\sum\limits_{k=0}^ma_kx^k:=\sum_{k=0}^{+\infty}b_kx^k, \end{align*}with $b_0=a_0,$ $b_k=\dfrac{g_{\la}^{(k-1)}(0)}{k!}+a_k$ for $0<k\leq m$ and $b_k=\dfrac{g_{\la}^{(k-1)}(0)}{k!}$ for $k> m.$ For $c_k=2^kk!k^2$ as in Lemma \ref{lemma6a}, we have\begin{align*}& \sum\limits_{k=0}^{\infty}|b_k|c_k^{\frac{1}{2}}\leq \sum\limits_{k=0}^{m}|a_k|c_k^{\frac{1}{2}}+\sum\limits_{k=1}^{+\infty}\dfrac{|g_{\la}^{(k-1)}(0)|}{k!}(2^k k!)^{\frac{1}{2}}k\\ &\leq \sum\limits_{k=0}^{m}|a_k|c_k^{\frac{1}{2}}+\sum\limits_{k=1}^{+\infty}2{\la^{k-1}}\|f_1'\|_{L^{\infty}(\R)}(2^k /k!)^{\frac{1}{2}}k<+\infty. \end{align*}Since $\mathcal L_n(F_{\la})-b_0=\sum\limits_{k=1}^\infty b_kL_n^{-1} \Tr H^k, $ we have\begin{align*}{n\choose q_n}\text{var}[\mathcal L_n(F_{\la})]={n\choose q_n}\sum\limits_{k=1}^{+\infty}\sum\limits_{k'=1}^{+\infty}b_kb_{k'}\text{cov}(L_n^{-1} \Tr H^k,L_n^{-1} \Tr H^{k'}).
\end{align*}
By Lemma \ref{lemma6a}, for the Gaussian case, we will have $${n\choose q_n}|\text{cov}(L_n^{-1} \Tr H^k,L_n^{-1} \Tr H^{k'})|\leq (c_kc_{k'})^{\frac{1}{2}},$$ where $c_k=2^kk!k^2.$ We notice that\begin{align*}\sum\limits_{k=1}^{+\infty}\sum\limits_{k'=1}^{+\infty}|b_kb_{k'}|(c_kc_{k'})^{\frac{1}{2}}=
\left(\sum\limits_{k=0}^{\infty}|b_k|c_k^{\frac{1}{2}}\right)^2<+\infty.
\end{align*} Therefore,  by the dominated convergence theorem and Lemma \ref{lemma11}, we have\begin{align*}&\lim\limits_{n\to+\infty}{n\choose q_n}\text{var}[\mathcal L_n(F_{\la})]=\lim\limits_{n\to+\infty}{n\choose q_n}\sum\limits_{k=1}^{+\infty}\sum\limits_{k'=1}^{+\infty}b_kb_{k'}\text{cov}(L_n^{-1} \Tr H^k,L_n^{-1} \Tr H^{k'})\\&=2\sum\limits_{k=1}^{+\infty}\sum\limits_{k'=1}^{+\infty}b_kb_{k'}(m^a_kk/2)(m^a_{k'}k'/2) 
=2\left(\sum\limits_{k=1}^{+\infty}b_km^a_kk/2\right)^2 \\&=2\left(\sum\limits_{k=1}^{+\infty}b_k\langle x^k,  \rho_{\infty}\rangle k/2\right)^2 =2\left\langle \sum\limits_{k=1}^{+\infty}b_kx^kk/2,  \rho_{\infty}\right\rangle^2 .
\end{align*}  Since $xF_{\la}'(x)/2=\sum\limits_{k=1}^{+\infty}b_kx^kk/2$, we finally have \begin{align}\label{v1}\lim\limits_{n\to+\infty}{n\choose q_n}\text{var}[\mathcal L_n(F_{\la})]=2\left\langle xF_{\la}'/2,  \rho_{\infty}\right\rangle^2.
\end{align}Since $(f-F_{\la})'=f'-F_{\la}'=f_1'+f_2'-(g_{\la}+f_2')=g-g_{\la}, $ by Lemma \ref{lemvar}, we have$${n\choose q_n} \text{var}[\mathcal L_n(f-F_{\la})]\leq C\|(f-F_{\la})'\|_{L^{\infty}(\R)}^2=C\|g-g_{\la}\|_{L^{\infty}(\R)}^2.$$ If we combine this with $$|(\text{var}[\mathcal L_n(f)])^{\frac{1}{2}}-(\text{var}[\mathcal L_n(F_{\la})])^{\frac{1}{2}}|\leq(\text{var}[\mathcal L_n(f-F_{\la})])^{\frac{1}{2}} ,$$ 
 we have, 
\begin{align*}{n\choose q_n}^{\frac{1}{2}}(\text{var}[\mathcal L_n(F_{\la})])^{\frac{1}{2}}-C\|g-g_{\la}\|_{L^{\infty}(\R)}\leq{n\choose q_n}^{\frac{1}{2}}(\text{var}[\mathcal L_n(f)])^{\frac{1}{2}}\\ \leq{n\choose q_n}^{\frac{1}{2}}(\text{var}[\mathcal L_n(F_{\la})])^{\frac{1}{2}}+C\|g-g_{\la}\|_{L^{\infty}(\R)}.
\end{align*}Taking limit, we have\begin{align*}\lim\limits_{n\to+\infty}{n\choose q_n}^{\frac{1}{2}}(\text{var}[\mathcal L_n(F_{\la})])^{\frac{1}{2}}-C\|g-g_{\la}\|_{L^{\infty}(\R)}\leq\liminf\limits_{n\to+\infty}{n\choose q_n}^{\frac{1}{2}}(\text{var}[\mathcal L_n(f)])^{\frac{1}{2}}\\ \leq\limsup\limits_{n\to+\infty}{n\choose q_n}^{\frac{1}{2}}(\text{var}[\mathcal L_n(f)])^{\frac{1}{2}}\leq\lim\limits_{n\to+\infty}{n\choose q_n}^{\frac{1}{2}}(\text{var}[\mathcal L_n(F_{\la})])^{\frac{1}{2}}+C\|g-g_{\la}\|_{L^{\infty}(\R)}.
\end{align*}Notice that (using \eqref{v1})\begin{align*}\lim\limits_{n\to+\infty}{n\choose q_n}^{\frac{1}{2}}(\text{var}[\mathcal L_n(F_{\la})])^{\frac{1}{2}}=2^{\frac{1}{2}}|\left\langle xF_{\la}'/2,  \rho_{\infty}\right\rangle|,
\end{align*}and that\begin{align*}\left||\left\langle xf'/2,  \rho_{\infty}\right\rangle|-|\left\langle xF_{\la}'/2,  \rho_{\infty}\right\rangle|\right|\leq|\left\langle x(f-F_{\la})'/2,  \rho_{\infty}\right\rangle|\\ \leq \|(f-F_{\la})'\|_{L^{\infty}(\R)}\left\langle |x|,  \rho_{\infty}\right\rangle/2\leq \|g-g_{\la}\|_{L^{\infty}(\R)}/2,
\end{align*}(here, we use the fact that $\left\langle 1,  \rho_{\infty}\right\rangle=\left\langle x^2,  \rho_{\infty}\right\rangle=1,\ \left\langle |x|,  \rho_{\infty}\right\rangle\leq \left\langle (1+x^2)/2,  \rho_{\infty}\right\rangle=1$ for the limiting densities in all cases) we will have\begin{align*}2^{\frac{1}{2}}|\left\langle xf'/2,  \rho_{\infty}\right\rangle|-C\|g-g_{\la}\|_{L^{\infty}(\R)}\leq\liminf\limits_{n\to+\infty}{n\choose q_n}^{\frac{1}{2}}(\text{var}[\mathcal L_n(f)])^{\frac{1}{2}}\\ \leq\limsup\limits_{n\to+\infty}{n\choose q_n}^{\frac{1}{2}}(\text{var}[\mathcal L_n(f)])^{\frac{1}{2}}\leq2^{\frac{1}{2}}|\left\langle xf'/2,  \rho_{\infty}\right\rangle|+C\|g-g_{\la}\|_{L^{\infty}(\R)}.
\end{align*}Since $ \lim\limits_{\la\to+\infty}\|g-g_{\la}\|_{L^{\infty}(\R)}=0$, letting $\la\to+\infty$, we finally have\begin{align*}\lim\limits_{n\to+\infty}{n\choose q_n}^{\frac{1}{2}}(\text{var}[\mathcal L_n(f)])^{\frac{1}{2}}=2^{\frac{1}{2}}|\left\langle xf'/2,  \rho_{\infty}\right\rangle|,
\end{align*} %and \begin{align*}\lim\limits_{n\to+\infty}{n\choose q_n}(\text{var}[\mathcal L_n(f)])=2\left\langle xf'/2,  \rho_{\infty}\right\rangle^2.
%\end{align*}
which completes the proof.\end{proof}

Now we can finish the proof of Theorem \ref{main3a}. Let $f(x)$ be Lipschitz and $f'(x)$ is bounded uniformly continuous. Let's denote $f_1:=f$ and $f_2:=-A x^2$ with $A:=\langle xf'/2,  \rho_{\infty}\rangle$ where it's easy to see  $|A|< \infty$. Then we have $ \langle x[f_1+f_2]'/2,  \rho_{\infty}\rangle=\langle xf'/2,  \rho_{\infty}\rangle-A\langle x^2,  \rho_{\infty}\rangle=\langle xf'/2,  \rho_{\infty}\rangle-A =0.$   Thus, by Lemma \ref{lemma12a}, we have\begin{align*}\lim\limits_{n\to+\infty}{n\choose q_n}\text{var}[\langle [f_1+f_2],   \rho_n\rangle]=2\left\langle x[f_1+f_2]'/2,  \rho_{\infty}\right\rangle^2=0.
\end{align*} Therefore,  we have $${n\choose q_n}^{\frac{1}{2}}\left (\langle [f_1+f_2],  \rho_n\rangle-\e[\langle [f_1+f_2],  \rho_n\rangle]\right)\to 0 $$ in probability. By definitions of $f, f_1, f_2$ and $\mu$ above,  we have \begin{align*}&{n\choose q_n}^{\frac{1}{2}} (\langle f,   \rho_n\rangle-\e[\langle f,   \rho_n\rangle])\\&={n\choose q_n}^{\frac{1}{2}} (\langle [f_1+f_2],   \rho_n\rangle-\e[\langle [f_1+f_2], \rho_n\rangle])+ A{n\choose q_n}^{\frac{1}{2}} (\langle x^2,   \rho_n\rangle-\e[\langle x^2, \rho_n\rangle]).\end{align*}  The first term tends to 0 in probability and Theorem \ref{main3} implies that the second term tends to $ A J$ in distribution where $J$ is the Gaussian distribution with mean 0 and variance 2, therefore, we have
%  $\\  {n\choose q_n}^{\frac{1}{2}} (\langle x^2,  \rho_n(\lambda)\rangle-\e[\langle x^2,   \rho_n(\lambda)\rangle])\Rightarrow J,$  and

$${n\choose q_n}^{\frac{1}{2}} (\langle f,  \rho_n\rangle-\e[\langle f,  \rho_n\rangle])\Rightarrow  A J,$$
which completes the proof of Theorem \ref{main3a}.
%\subsection{Proof of Berry-Esseen Theorem}
%\begin{thm}If $J_{i_1\cdots i_{q_n}}$ are independent, standard Gaussian random variables, by Lemma \ref{lemma12a} we know that Theorem \ref{main3} is still true for $f:\R\to\R$ Lipschitz such that $f'$ is bounded uniformly continuous.\end{thm}

\end{document}